\documentclass[journal,onecolumn]{IEEEtran}
\usepackage{amsmath,amsfonts,amssymb,amsthm,mathtools}

\usepackage{algorithm}
\usepackage{algpseudocode}
\usepackage{array}
\usepackage[caption=false,font=normalsize,labelfont=sf,textfont=sf]{subfig}
\usepackage{textcomp}
\usepackage{stfloats}
\usepackage{url}
\usepackage{multirow}
\usepackage{verbatim}
\usepackage{graphicx}
\usepackage{cite}
\usepackage{geometry}
\usepackage{float}
\usepackage{enumitem}
\usepackage{tabularx, booktabs, ragged2e}
\newcolumntype{L}{>{\RaggedRight\arraybackslash}X}
\usepackage{placeins}
\usepackage{makecell}
\usepackage{etoolbox} 
\newcommand{\Dd}{\text{Dd}}
\newcommand{\Atom}{\text{Atom}}

\newcommand{\Cn}{\text{Cn}}

\newcommand{\E}{\mathbb{E}}
\newcommand{\Ess}{\operatorname{Ess}}

\newtheorem{assumption}{Assumption}[section]
\newtheorem{theorem}{Theorem}[section]
\newtheorem{lemma}{Lemma}[section]
\newtheorem{proposition}{Proposition}[section]
\newtheorem{corollary}{Corollary}[section]
\newtheorem{definition}{Definition}[section]
\newtheorem{remark}{Remark}[section]
\newtheorem{example}{Example}[section]

\numberwithin{equation}{section}

\makeatletter
\renewenvironment{proof}[1][\proofname]{\par
  \pushQED{\qed}%
  \normalfont \topsep6\p@\@plus6\p@\relax
  \trivlist
  \item[\hskip\labelsep
        \itshape
    #1\@addpunct{.}]\ignorespaces
}{%
  \popQED\endtrivlist\@endpefalse
}
\makeatother

\DeclareMathOperator{\Ind}{Ind}

\begin{document}

\title{Rate-Distortion Theory for Deductive Sources under Closure Fidelity}

\author{Jianfeng~Xu$^{1}$%
\thanks{$^{1}$Koguan School of Law, China Institute for Smart Justice,
School of Computer Science, Shanghai Jiao Tong University,
Shanghai 200030, China. Email: xujf@sjtu.edu.cn}%
}

\maketitle

\begin{abstract}
We study lossy compression of a finite statement source generated in a fixed deductive environment.
The source symbols are statements in a knowledge base endowed with a shared proof system, and reconstruction fidelity is measured by preservation of deductive closure rather than by symbolwise equality.
Fixing the proof system and a canonical scan order yields a decomposition of the source alphabet into an irredundant core and redundant stored consequences.
At zero distortion, each core symbol induces a set of distortion-free reconstructions.
In the nonconfusable (disjoint-core) regime, we show that the minimum zero-distortion rate equals the source mass of the core times the entropy of the source conditioned on that core.
In the general confusable-core regime, we characterise the exact zero-distortion rate via a hypergraph-entropy quantity induced by jointly realisable core subsets, with a reduction to K\"orner-style graph entropy under a natural pairwise realisability condition.
For reconstruction alphabets contained in the deductive closure of the source knowledge base, we further prove that the full rate--distortion function depends only on the core, so redundant states are invisible to both rate and distortion.
Finally, when the decoder is limited to a bounded inference-depth budget (a bounded number of iterations of the immediate-consequence operator), we obtain an exact rate--depth--distortion characterisation.
Under an additional order-robustness assumption identifying the chosen core with the order-free essential set, this characterisation interpolates between classical symbolwise compression and unconstrained deductive compression.
\end{abstract}


\section{Introduction}
\label{sec:introduction}

\IEEEPARstart{C}{lassical} rate--distortion theory studies the minimum description rate required to represent a random source under a prescribed fidelity criterion \cite{shannon1948mathematical,cover2006elements,csiszar2011information}. In the standard formulation, the source alphabet is treated as a collection of symbols and distortion is imposed directly on source--reconstruction pairs. In many modern communication settings, however, the source is not naturally a flat alphabet but a finite knowledge base equipped with a shared proof system. In such settings, symbolwise reproduction can be unnecessarily stringent: if the receiver can reconstruct the same deductive consequences from a different stored representation, then the communication task has arguably been completed.

This observation suggests a source-coding problem in which fidelity is measured not by exact symbol reproduction but by preservation of deductive closure. We study this problem for a finite source set \(S_O\) endowed with a fixed proof system and its induced closure operator \(\Cn(\cdot)\). The resulting source model carries an intrinsic redundancy structure: some states are irredundant premises, while others are stored consequences that can be re-derived once those premises are known. Once the proof system and canonical order are fixed, this induces a
decomposition of the source into an irredundant core
\(A:=\Atom(S_O)\) and a redundant part \(J:=S_O\setminus A\). Under closure fidelity, redundant states need not be transmitted faithfully, so the effective source rate can be strictly smaller than the classical rate.

Operationally, the source studied here is not a random knowledge base
but a \emph{statement-level source generated within a fixed deductive
environment}. A finite knowledge base \(S_O\), together with its proof
system and closure operator, is fixed as part of the model, and the
random source symbol is a statement \({\mathsf{S}}_o\sim P_O\) drawn from that
environment. Thus a source block models a stream of statement reports,
updates, query answers, or semantic tokens produced relative to the
same local knowledge base. This is why the distortion may depend on
the ambient \(S_O\) while remaining a standard single-letter
distortion function.

The present formulation is related to, but distinct from, the broader
semantic-communication and task-oriented communication literature,
including early conceptual proposals\cite{carnap1952outline}, modern surveys\cite{gunduz2022beyond}, and recent
mathematical formulations of semantic rate--distortion and semantic
coding
\cite{liu2022indirect,niu2024mathematical,zhao2025semantic,
han2025extended,ma2025theory}.
Our formulation is narrower and more Shannon-theoretic: we fix a
finite deductive environment and let the proof system itself induce
zero-distortion equivalence through closure preservation.

\paragraph*{Related work and positioning.}
Two classical information-theoretic lines are especially relevant.
First, source coding with decoder side information and
computation-oriented compression---from Slepian--Wolf and Wyner--Ziv
to coding for computing---show how decoder-side structure or
non-identity reconstruction goals can change fundamental compression
limits \cite{slepian1973noiseless,wyner1976rate,orlitsky2001coding}.
Our setting differs in that the decoder's advantage is not exogenous
random side information and the reproduction target is not a fixed
function specified in advance; instead, the shared proof system
induces the fidelity criterion itself.

Second, the confusable-core phenomenon that appears here is closely
connected to zero-error and graph-theoretic information theory.
Shannon's zero-error framework introduced confusability as an
information-theoretic object, K\"orner's graph entropy and
Witsenhausen's zero-error side-information formulation linked source
coding to graph structure, and later work further clarified these
connections \cite{shannon1956zero,korner1986fredman,witsenhausen1976zero,alon1996source}.
When Assumption~\ref{assump:core-disjoint} fails, our zero-distortion
problem becomes a genuine confusability problem on the irredundant core.
In this paper we resolve the resulting \(D=0\) limit exactly via a
hypergraph-entropy characterization (Section~\ref{sec:confusable}),
with a reduction to K\"orner-style graph entropy under a natural
pairwise realisability condition \cite{korner1971coding,alon1996source,orlitsky2001coding}.

A further adjacent line comes from logic and database theory,
especially the fixpoint semantics of Horn programs, finite-domain
Datalog, reasoning over materialized deductive stores, and the broader
theory of incomplete information and certain answers
\cite{van1976semantics,ceri1989you,abiteboul1995foundations,
dantsin2001complexity,hu2019datalog,gupta1995maintenance,
gupta1998materialized}.
Our redundant/materialized split is closest to the materialized-view
perspective: the set \(J\) behaves as stored consequences that need
not be communicated once the irredundant premises are available. An
incomplete-information analogue would replace deductive closure by a
certain-answer operator, leading to a semantics-based fidelity notion.
We do not pursue that extension here; the present paper remains a
Shannon-theoretic study of a fixed deductive environment.

Finally, our work is also adjacent to recent rate--distortion
formulations for semantic sources and to perception-constrained source
coding
\cite{liu2022indirect,zhao2025semantic,han2025extended,
salehkalaiabar2024rate}. These works broaden the classical fidelity
framework beyond symbolwise criteria, but they do not model a fixed
deductive environment or exploit the proof-induced decomposition into
an irredundant core and redundant stored consequences that is central
to the present deductive-source model.

This paper asks three basic information-theoretic questions. First, what is the exact zero-distortion rate under closure fidelity? Second, how does the full rate--distortion function depend on the irredundant core? Third, how does this picture change when the decoder has only a limited
inference-depth budget (a limited number of iterations)? We answer these questions in a finite-alphabet setting.
A transparent Shannon-entropy closed form for the zero-distortion limit requires a disjointness (nonconfusability) assumption on core reconstruction sets; we state this assumption up front.
When disjointness fails, the exact zero-distortion limit remains characterisable, but the answer becomes combinatorial and is expressed in terms of hypergraph/graph entropy (Section~\ref{sec:confusable}).

Our main message is that closure fidelity converts deductive redundancy into rate savings. In particular, the zero-distortion problem reduces to coding only the irredundant core, and the full rate--distortion function decomposes into a contribution from that core alone. We further show that when the decoder can perform at most \(\delta\)
iterations of the immediate-consequence operator, the relevant effective source becomes a
\(\delta\)-irredundant core. Under an additional order-robustness
assumption aligning the chosen core with the order-free essential
set, this yields an exact rate--depth--distortion interpolation
between classical symbolwise compression and unconstrained
deductive compression. Source--channel separation, closure-adapted converse bounds, and heterogeneous-receiver extensions then follow as consequences of this source-model characterization.

\subsection*{Main Contributions}

\begin{itemize}

\item \textbf{A source-model formulation for deductive compression.}
We formulate a finite-alphabet rate--distortion problem in which the source is a knowledge base endowed with a fixed proof system and the fidelity criterion is preservation of deductive closure. This shifts the emphasis from scheme-specific semantic communication architectures to a source-model-level information-theoretic object.

\item \textbf{An exact zero-distortion characterization under closure fidelity.}
Under Assumption~\ref{assump:core-disjoint}, we prove that the minimum zero-distortion rate is
\[
R_{\mathrm{sem}}(0)=P_A\,H(\pi_A),
\]
where \(A=\Atom(S_O)\) is the irredundant core, \(P_A\) is its source mass, and \(\pi_A\) is the source distribution conditioned on \(A\) \textup{(Theorem~\ref{thm:tight-zero-rate})}. Thus the effective zero-distortion rate is determined by the core rather than by the full stored source.

\item \textbf{A hypergraph/graph-entropy law for the confusable-core zero-distortion limit.}
Dropping Assumption~\ref{assump:core-disjoint}, we characterise the exact zero-distortion rate via a hypergraph entropy induced by jointly realisable core subsets:
\[
R_{\mathrm{sem}}(0)=P_A\,H_{\Gamma_0}(\pi_A)
\]
(Section~\ref{sec:confusable}).
Under a natural pairwise realisability condition, this reduces to a K\"orner-style graph-entropy expression.

\item \textbf{A core-only decomposition of the full rate--distortion function.}
For reconstruction alphabets contained in \(\Cn(S_O)\), we show that
the entire rate--distortion function decomposes into a contribution
from the irredundant core alone
\textup{(Theorem~\ref{thm:rd-decomposition})}. In this sense,
redundant states are invisible to both rate and distortion.

\item \textbf{An exact rate--depth--distortion law for limited inference.}
When the decoder is allowed only a bounded inference-depth budget,
we introduce the \(\delta\)-irredundant core and prove the exact
zero-distortion rate--depth law
\[
R_{\mathrm{sem}}(0,\delta)=P_\delta\,H(\pi_\delta),
\]
together with the corresponding full rate--depth--distortion function
\(R_{\mathrm{sem}}(D,\delta)\) over \((D,\delta)\)
\textup{(Theorems~\ref{thm:rate-depth} and~\ref{thm:rate-depth-distortion})}.
Under an additional order-robustness assumption ensuring
\(\Atom(S_O)=\Ess(S_O)\), this yields a sharp interpolation between
classical symbolwise coding and unconstrained deductive coding.

\item \textbf{Information-theoretic consequences and extensions.}
The source characterization leads to source--channel separation statements, closure-adapted converse bounds of Fano type, and extensions to heterogeneous receivers with mismatched vocabularies. These results show how shared inference structure affects reliable communication beyond the homogeneous source-coding setting.

\end{itemize}

The remainder of the paper is organized as follows.
Section~\ref{sec:deductive-source} introduces the deductive source
model and closure fidelity.
Section~\ref{sec:exact-zero-rate} proves the exact zero-distortion
rate in the nonconfusable (disjoint-core) regime.
Section~\ref{sec:confusable} drops disjointness and characterises the
exact zero-distortion rate via hypergraph entropy, with a graph-entropy
reduction under pairwise realisability.
Section~\ref{sec:core-decomposition} derives the full
rate--distortion decomposition over the irredundant core.
Section~\ref{sec:limited-inference} studies bounded inference and
the resulting rate--depth--distortion tradeoff.
Section~\ref{sec:consequences} presents consequences including
source--channel separation, converse bounds, and a heterogeneous
extension.
Appendix~\ref{app:proofs} collects deferred proofs of the main
results in Sections~III--VII.
Appendix~\ref{app:datalog-classes} records concrete Datalog
classes satisfying the standing assumptions.
Appendix~\ref{app:extra-examples} contains additional examples and
counterexamples.

\section{Deductive Source Model and Closure Fidelity}
\label{sec:deductive-source}

This section isolates the minimum structure needed for the
source-coding problem studied in this paper.
At this stage no channel model is required.
The basic object is a finite source alphabet of statements,
equipped with a fixed proof system and its induced deductive
closure operator.
This structure induces, once the proof system and canonical order are
fixed, a decomposition of the source into irredundant premises and
redundant stored consequences, which in turn determines the
closure-based fidelity criterion.

\subsection{Deductive Source and Irredundant Core}
\label{subsec:deductive-source-model}

\begin{assumption}[Fixed effective proof system]
\label{assump:proof-system}
We fix a finite ambient universe \(\mathbb{S}_O\) of candidate
statements together with an effective proof system
\(\mathsf{PS}\).
Hence \(\Cn(\Gamma)\subseteq \mathbb{S}_O\) is finite for every finite \(\Gamma\).
For every finite \(\Gamma\subseteq\mathbb{S}_O\) and
every \(s\in\mathbb{S}_O\), whether
\(\Gamma\vdash_{\mathrm{kb}} s\) is decidable.
The induced deductive closure operator is
\[
  \Cn(\Gamma)
  \;:=\;
  \{s\in\mathbb{S}_O:\Gamma\vdash_{\mathrm{kb}} s\},
\]
viewed as the standard logical closure operator associated with the
proof system
\cite{hodges1993model}.
It satisfies:
\begin{enumerate}[label=\textup{(Cn\arabic*)}]
  \item \emph{Reflexivity:} \(\Gamma\subseteq\Cn(\Gamma)\).
  \item \emph{Monotonicity:}
        \(\Gamma\subseteq\Gamma' \Rightarrow
        \Cn(\Gamma)\subseteq\Cn(\Gamma')\).
  \item \emph{Idempotence:}
        \(\Cn(\Cn(\Gamma))=\Cn(\Gamma)\).
\end{enumerate}
\end{assumption}

\begin{assumption}[Finite source alphabet]
\label{assump:finite-so}
The transmitted knowledge base \(S_O\subseteq\mathbb{S}_O\)
is finite and effectively listable under a fixed canonical order.
The reconstruction alphabet \(\hat S_O\subseteq\mathbb{S}_O\)
is also finite and nonempty.
\end{assumption}

\begin{assumption}[Decidable redundancy test]
\label{assump:core-extractable}
For every finite \(\Gamma\subseteq S_O\) and every \(s\in S_O\),
membership \(s\in\Cn(\Gamma)\) is decidable.
\end{assumption}

\begin{remark}[On Assumption~\ref{assump:core-extractable}]
Assumption~\ref{assump:core-extractable} is implied by
Assumption~\ref{assump:proof-system}; it is stated separately only to
highlight the specific decidability query used in the deletion
procedure of Definition~\ref{def:atom-so}.
\end{remark}

\begin{definition}[Deductive source]
\label{def:deductive-source}
A \emph{deductive source} is a pair \((S_O,P_O)\), where
\(S_O\subseteq\mathbb{S}_O\) is a finite knowledge base and
\(P_O\in\Delta(S_O)\) is a probability distribution on \(S_O\),
together with the fixed proof system of
Assumption~\ref{assump:proof-system}.
\end{definition}

\begin{definition}[Irredundant core]
\label{def:atom-so}
Let \(S_O\) be a finite knowledge base.
Define \(\Atom(S_O)\) by the following deterministic deletion
procedure:
initialize \(A\gets S_O\); scan the elements of \(S_O\) in the
fixed canonical order; whenever the current element \(s\) satisfies
\(s\in\Cn(A\setminus\{s\})\), remove it from \(A\).
The final set is called the \emph{irredundant core} of \(S_O\).
\end{definition}

\begin{definition}[Core and redundant part]
\label{def:intrinsic-operational-bases}
For a deductive source \(S_O\), write
\[
  A:=\Atom(S_O),
  \qquad
  J:=S_O\setminus A.
\]
The set \(A\) is the \emph{irredundant core}, and \(J\) is the
\emph{redundant part}.
\end{definition}

\begin{proposition}[Core correctness]
\label{prop:atom-core-correct}
Under
Assumptions~\textup{\ref{assump:proof-system}}--\textup{\ref{assump:core-extractable}},
the set \(A:=\Atom(S_O)\) satisfies:
\begin{enumerate}[label=\textup{(\roman*)}]
  \item \(\Cn(A)=\Cn(S_O)\);
  \item for every \(a\in A\),
        \(a\notin\Cn(A\setminus\{a\})\)
        \textup{(irredundancy)};
  \item \(\Atom(S_O)\) is uniquely determined by
        \(S_O\), \(\Cn\), and the fixed canonical order.
\end{enumerate}
\end{proposition}

\begin{proof}
Each deletion preserves closure.
Indeed, if \(s\in\Cn(A\setminus\{s\})\), then
\(A\subseteq\Cn(A\setminus\{s\})\), hence by monotonicity and
idempotence,
\[
  \Cn(A)\subseteq\Cn(\Cn(A\setminus\{s\}))
  =\Cn(A\setminus\{s\}),
\]
while the reverse inclusion is immediate from
\(A\setminus\{s\}\subseteq A\).
By induction over the scan, the final output satisfies
\(\Cn(A)=\Cn(S_O)\), proving~\textup{(i)}.

For~\textup{(ii)}, let \(a\in A\) and let
\(A_{\mathrm{scan}}\) denote the current set when \(a\) is scanned.
Since \(a\) is retained,
\(a\notin\Cn(A_{\mathrm{scan}}\setminus\{a\})\).
Later deletions only shrink the set, so
\(A\setminus\{a\}\subseteq A_{\mathrm{scan}}\setminus\{a\}\).
Monotonicity then gives
\[
  \Cn(A\setminus\{a\})
  \subseteq
  \Cn(A_{\mathrm{scan}}\setminus\{a\}),
\]
hence \(a\notin\Cn(A\setminus\{a\})\).

Part~\textup{(iii)} follows from the determinism of the deletion
procedure once \(\Cn\) and the canonical order are fixed.
\end{proof}

\begin{remark}[Proof-system dependence and what is intrinsic]
\label{rem:ps-dependence}
The set \(\Atom(S_O)\) is relative to the chosen proof system and, in
general, to the canonical order.
This dependence is not a defect of the model: it reflects the fact
that communication value is measured relative to shared inference
capability.
What is primary in this paper is the induced
rate--distortion quantity, not any one particular realizing core.
Section~\ref{sec:exact-zero-rate} shows that the scalar
zero-distortion rate is intrinsic to the deductive source model, even
when different canonical orders may produce different irredundant
cores.
Appendix~\ref{app:order-invariance-positive} then gives a sufficient
condition under which the core itself becomes order-invariant.
\end{remark}

\subsection{Closure Fidelity and Closure Distortion}
\label{subsec:closure-fidelity}

The source model above suggests a fidelity criterion stronger than
semantic similarity but weaker than symbolwise equality:
a reconstruction should be considered faithful whenever it
preserves the deductive consequences of the transmitted knowledge
base.

\begin{definition}[Closure fidelity]
\label{def:closure-fidelity}
For finite \(S,\hat S\subseteq\mathbb{S}_O\), define
\[
  \mathsf{F}_{\Cn}(S,\hat S)
  \;:=\;
  \frac{|\Cn(S)\cap\Cn(\hat S)|}
       {|\Cn(S)\cup\Cn(\hat S)|},
\]
with \(0/0:=1\).
Thus \(\mathsf{F}_{\Cn}(S,\hat S)=1\) if and only if
\(\Cn(S)=\Cn(\hat S)\).
\end{definition}

\begin{remark}[Closure fidelity as a Jaccard-type similarity]
\label{rem:closure-jaccard}
The quantity \(\mathsf{F}_{\Cn}(S,\hat S)\) is the Jaccard similarity
between the two closure sets \(\Cn(S)\) and \(\Cn(\hat S)\). Jaccard-
type similarity and the associated Tanimoto/Jaccard distance are
classical set-comparison tools; see, e.g.,
\cite{broder1997resemblance,lipkus1999proof}. Here the comparison is
lifted from raw symbol sets to deductive closures.
\end{remark}

\begin{remark}[Conceptual relation to semantic content]
\label{rem:carnap-bar-hillel}
The closure-based viewpoint is conceptually related to the early
Carnap--Bar-Hillel account of semantic information, in which the
informational content of a sentence is tied to the set of possibilities
it excludes \cite{carnap1952outline}. Our quantity
\(\mathsf{F}_{\Cn}(S,\hat S)\) is not a possible-world measure; it is
a proof-theoretic comparison of deductive consequence sets. Still, as
a Jaccard-type similarity on closures, it can be viewed as a natural
way to quantify how much inferential content is preserved by a
reconstruction.
\end{remark}

\begin{definition}[Hamming distortion]
\label{def:hamming-distortion}
The Hamming distortion \cite[Ch.~10]{cover2006elements} is
\[
  d_H(s_o,\hat s):=\mathbf{1}[s_o\neq \hat s].
\]
\end{definition}

\begin{definition}[Closure distortion]
\label{def:closure-distortion}
Fix a reference knowledge base \(\Gamma\subseteq\mathbb{S}_O\)
and write \(\Gamma_{-s}:=\Gamma\setminus\{s\}\).
For \(s_o\in\Gamma\) and \(\hat s\in\hat S_O\), define
\[
  C_{s_o}:=\Cn(\Gamma_{-s_o}\cup\{s_o\}),
  \qquad
  C_{\hat s}:=\Cn(\Gamma_{-s_o}\cup\{\hat s\}),
\]
and set
\begin{equation}\label{eq:d-Cn}
  d_{\Cn}(s_o,\hat s\mid\Gamma)
  :=
  1-\frac{|C_{s_o}\cap C_{\hat s}|}
          {|C_{s_o}\cup C_{\hat s}|},
\end{equation}
with \(0/0:=0\), i.e., zero distortion when both closure sets are empty.
In the sequel we use the single-letter distortion
\(d_{\Cn}(s_o,\hat s):=d_{\Cn}(s_o,\hat s\mid S_O)\).
\end{definition}

\begin{remark}[Why closure distortion is the right object here]
\label{rem:closure-distortion-content}
Under \(d_H\), every symbol error is costly.
Under \(d_{\Cn}\), replacing a statement by a different one may be
free if the deductive closure is preserved.
This is precisely the mechanism by which deductive redundancy
becomes compressible.
The entire zero-distortion theory developed in the next section
rests on this distinction.
\end{remark}

\begin{proposition}[Redundant states are zero-distortion free]
\label{prop:redundant-free}
Let \(A=\Atom(S_O)\) and \(J=S_O\setminus A\).
If \(j\in J\) and \(\hat s\in\Cn(S_O)\), then
\[
  d_{\Cn}(j,\hat s\mid S_O)=0.
\]
\end{proposition}

\begin{proof}
Because \(A=\Atom(S_O)\subseteq S_O\setminus\{j\}\), one has by
Proposition~\ref{prop:atom-core-correct}(i) and monotonicity that
\[
  \Cn(S_O)
  =
  \Cn(A)
  \subseteq
  \Cn(S_O\setminus\{j\}).
\]
Since also \(S_O\setminus\{j\}\subseteq S_O\), the reverse inclusion
\(\Cn(S_O\setminus\{j\})\subseteq \Cn(S_O)\) is immediate.
Hence
\[
  \Cn(S_O\setminus\{j\})=\Cn(S_O).
\]

Because \(\hat s\in\Cn(S_O)=\Cn(S_O\setminus\{j\})\), one has
\[
  (S_O\setminus\{j\})\cup\{\hat s\}
  \subseteq
  \Cn(S_O\setminus\{j\}).
\]
By monotonicity and idempotence,
\[
  \Cn\bigl((S_O\setminus\{j\})\cup\{\hat s\}\bigr)
  =
  \Cn(S_O\setminus\{j\})
  =
  \Cn(S_O).
\]
Hence the two closures in
Definition~\ref{def:closure-distortion} coincide, and the
Jaccard distance is zero.
\end{proof}

\begin{definition}[Closure-based rate--distortion function]
\label{def:semantic-rd}
Let \((S_O,P_O)\) be a deductive source. Let \(\mathsf{S}_o\sim P_O\) and let
\(\hat{\mathsf{S}}_o\in \hat S_O\) be generated via a test channel
\(P_{\hat S\mid S}\) as \(\hat{\mathsf{S}}_o\sim P_{\hat S\mid S}(\cdot\mid \mathsf{S}_o)\).
With the single-letter distortion \(d_{\Cn}(\cdot,\cdot\mid S_O)\), define
\begin{equation}\label{eq:R-sem-D}
  R_{\mathrm{sem}}(D;\,d_{\Cn},\,P_O)
  :=
  \min_{\substack{P_{\hat S\mid S}:\\
  \E[d_{\Cn}(\mathsf{S}_o,\hat{\mathsf{S}}_o)]\le D}}
  I(\mathsf{S}_o;\hat{\mathsf{S}}_o).
\end{equation}
\end{definition}

\begin{remark}[Operational meaning of the source model]
\label{rem:fixed-kb-parameter}
The object being compressed in this paper is a \emph{statement stream
drawn from a fixed deductive environment}, not a random knowledge base.
More precisely, a finite knowledge base \(S_O\) is fixed once and for
all as part of the source model, together with its proof system and
closure operator, and the random source symbol \({\mathsf{S}}_o\) is drawn from
\(P_O\) on \(S_O\).

Operationally, one may think of \({\mathsf{S}}_o\) as a statement report, update,
query answer, or symbolic semantic token produced within the same local
knowledge environment \(S_O\).
A block of length \(m\) therefore models \(m\) independent draws from
the same deductive source environment.
Accordingly, the distortion
\(d_{\Cn}(s_o,\hat s\mid S_O)\) is a standard single-letter
distortion function parameterized by the fixed ambient knowledge base
\(S_O\).

We do \emph{not} study in this paper the distinct problem in which the
knowledge base itself is random across blocks.
That would lead to a higher-level hierarchical source model and is left
for future work.
\end{remark}

\begin{remark}[Roadmap]
\label{rem:section2-roadmap}
Section~\ref{sec:exact-zero-rate} studies the case \(D=0\) and
shows that the corresponding rate depends only on the irredundant
core.
Section~\ref{sec:limited-inference} later replaces the
unconstrained closure criterion by a bounded-inference version,
yielding a rate--depth--distortion tradeoff.
\end{remark}

\section{Exact Zero-Distortion Rate}
\label{sec:exact-zero-rate}

This section gives the main zero-distortion characterization
behind the paper.
Throughout, let
\[
  A:=\Atom(S_O),
  \qquad
  J:=S_O\setminus A,
  \qquad
  P_A:=P_O(A),
  \qquad
  P_J:=1-P_A,
\]
and, whenever \(P_A>0\), let
\[
  \pi_A(a):=\frac{P_O(a)}{P_A},
  \qquad a\in A.
\]
We write \(h_b(p):=-p\log p-(1-p)\log(1-p)\) for the binary entropy function,
and adopt the convention \(P_A H(\pi_A):=0\) when \(P_A=0\).

The key technical issue is whether different core symbols can share
the same zero-distortion reconstruction.
When they cannot, the zero-distortion problem admits an exact
single-letter solution.

\begin{assumption}[Disjoint zero-distortion reconstruction sets]
\label{assump:core-disjoint}
For each \(s\in S_O\), define its zero-distortion reconstruction set
\[
  R_{\Cn}(s)
  :=
  \{\hat s\in\hat S_O:\ d_{\Cn}(s,\hat s\mid S_O)=0\}.
\]
We assume that for distinct \(a_1,a_2\in A\),
\[
  R_{\Cn}(a_1)\cap R_{\Cn}(a_2)=\varnothing.
\]
\end{assumption}

\begin{remark}[Fully redundant degenerate case]
\label{rem:fully-redundant-case}
If \(A=\varnothing\) and there exists
\(\hat s^\star\in \hat S_O\cap\Cn(S_O)\), then the constant
reconstruction \(\hat {\mathsf{S}}_o\equiv \hat s^\star\) achieves zero closure
distortion for every source symbol, and hence
\[
  R_{\mathrm{sem}}(0;\,d_{\Cn},\,P_O)=0.
\]
The exact theorem below therefore focuses on the nondegenerate case
\(A\neq\varnothing\), where the core carries positive effective
complexity.
\end{remark}

\begin{remark}[Zero-mass core case]
\label{rem:zero-mass-core-case}
If \(A\neq\varnothing\), \(P_A=0\), and \(A\subseteq\hat S_O\), then
the source is supported on \(J=S_O\setminus A\). Hence any fixed
\(a_0\in A\) used as a constant reconstruction achieves zero closure
distortion almost surely, because \(a_0\in A\subseteq \Cn(S_O)\) and
every symbol in \(J\) is zero-distortion free against every
reconstruction in \(\Cn(S_O)\) by
Proposition~\ref{prop:redundant-free}. Therefore
\[
  R_{\mathrm{sem}}(0;\,d_{\Cn},\,P_O)=0.
\]
Accordingly, the only nontrivial case in
Theorem~\ref{thm:tight-zero-rate} is \(P_A>0\).
\end{remark}

\begin{remark}[Interpretation and naturality of
Assumption~\ref{assump:core-disjoint}]
\label{rem:core-disjoint-general}
Assumption~\ref{assump:core-disjoint} is a
\emph{zero-distortion nonconfusability condition} on the irredundant
core.
It is the closure-fidelity analogue of requiring pairwise disjoint
zero-error decoding regions: a reconstruction that is closure-correct
for one core symbol cannot simultaneously be closure-correct for a
different core symbol.

Thus the assumption is not an arbitrary technical restriction but the
precise condition under which the exact zero-distortion converse
collapses to a single-letter entropy expression.
When it fails, and still assuming \(A\subseteq\hat S_O\), the
achievability bound
\(R_{\mathrm{sem}}(0)\le P_AH(\pi_A)\) still holds, but the converse
must account for overlap among the sets
\(\{R_{\Cn}(a):a\in A\}\), leading to a genuine confusability problem
on the core.
Appendix~\ref{app:datalog-classes} gives concrete Datalog classes in
which the assumption is satisfied by construction.
\end{remark}

\begin{proposition}[Zero-distortion achievability]
\label{prop:zero-rate-ach}
Assume \(P_A>0\) and \(A\subseteq \hat S_O\). Then
\begin{equation}\label{eq:zero-rate-upper}
  R_{\mathrm{sem}}(0;\,d_{\Cn},\,P_O)
  \;\le\;
  P_A\,H(\pi_A).
\end{equation}
\end{proposition}

\begin{proof}
Define the conditional distribution
\[
  P_{\hat S\mid S}(\hat s\mid s_o)
  :=
  \begin{cases}
    \mathbf{1}[\hat s=s_o], & s_o\in A,\\[3pt]
    \pi_A(\hat s), & s_o\in J,
  \end{cases}
\]
where \(\pi_A\) is supported on \(A\subseteq\hat S_O\).

For \(a\in A\), the reconstruction is exact, so
\(d_{\Cn}(a,a\mid S_O)=0\).
For \(j\in J\), the reconstruction lies in
\(A\subseteq\Cn(S_O)\), hence
\(d_{\Cn}(j,\hat s\mid S_O)=0\) by
Proposition~\ref{prop:redundant-free}.
Therefore the expected distortion is zero.

Let \(T\) be the auxiliary variable
\[
  T(s_o)
  :=
  \begin{cases}
    s_o, & s_o\in A,\\
    *,   & s_o\in J.
  \end{cases}
\]
Under the above conditional law, \(\hat {\mathsf{S}}_o\) depends on \({\mathsf{S}}_o\)
only through \(T\), and the output marginal is exactly \(\pi_A\).
Hence
\[
  I({\mathsf{S}}_o;\hat {\mathsf{S}}_o)
  =I(T;\hat {\mathsf{S}}_o)
  =H(\hat {\mathsf{S}}_o)-H(\hat {\mathsf{S}}_o\mid T)
  =H(\pi_A)-P_J H(\pi_A)
  =P_A\,H(\pi_A).
\]
This proves~\eqref{eq:zero-rate-upper}.
\end{proof}

\begin{theorem}[Exact zero-distortion rate]
\label{thm:tight-zero-rate}
Assume \(A\neq\varnothing\), \(A\subseteq\hat S_O\), and
Assumption~\ref{assump:core-disjoint}.
Then the closure-based zero-distortion rate is exactly
\begin{equation}\label{eq:tight-R0}
  R_{\mathrm{sem}}(0;\,d_{\Cn},\,P_O)
  \;=\;
  P_A\,H(\pi_A),
\end{equation}
with the convention that the right-hand side is \(0\) when \(P_A=0\).
\end{theorem}

\begin{proof}[Proof sketch]
If \(P_A=0\), the claim follows from
Remark~\ref{rem:zero-mass-core-case}. Assume henceforth that
\(P_A>0\).

The achievability part is Proposition~\ref{prop:zero-rate-ach}.
For the converse, collapse the source to the auxiliary variable
\[
  T(s_o)
  :=
  \begin{cases}
    s_o, & s_o\in A,\\
    *,   & s_o\in J.
  \end{cases}
\]
Then \(I({\mathsf{S}}_o;\hat {\mathsf{S}}_o)\ge I(T;\hat {\mathsf{S}}_o)\) by data processing.
Under zero closure distortion, each core symbol \(a\in A\) must be
reconstructed inside \(R_{\Cn}(a)\), and
Assumption~\ref{assump:core-disjoint} makes these supports pairwise
disjoint; symbols in \(J\) impose no identifying constraint by
Proposition~\ref{prop:redundant-free}. A disjoint-support entropy
decomposition, combined with concavity of entropy, yields
\[
  I(T;\hat {\mathsf{S}}_o)\ge P_A H(\pi_A).
\]
Together with Proposition~\ref{prop:zero-rate-ach}, this proves
\eqref{eq:tight-R0}. Full details are given in
Appendix~\ref{app:proof-tight-zero-rate}.
\end{proof}

\begin{remark}[Canonical-order dependence versus intrinsic rate]
\label{rem:order-vs-rate}
The set \(A=\Atom(S_O)\) may depend on the fixed canonical order used
in Definition~\ref{def:atom-so}. This does \emph{not} mean that the
zero-distortion rate itself is order-dependent.
Indeed, the left-hand side of~\eqref{eq:tight-R0},
\(R_{\mathrm{sem}}(0;\,d_{\Cn},\,P_O)\), is defined entirely in terms
of the source law \(P_O\) and the closure distortion \(d_{\Cn}\), and
therefore is intrinsic to the deductive source model.
Consequently, for any canonical order under which
Theorem~\ref{thm:tight-zero-rate} applies, the scalar
\(P_AH(\pi_A)\) must equal the same intrinsic zero-distortion rate.

What may depend on the canonical order is the particular
core/redundant decomposition that realizes this value.
Appendix~\ref{app:order-invariance-positive} gives a sufficient
condition under which the core itself becomes order-invariant.
\end{remark}

\begin{corollary}[Uniform source and strict gain over Hamming fidelity]
\label{cor:zero-rate-gain}
Assume \(A\neq\varnothing\), and let \(k:=|A|\).
Under the hypotheses of
Theorem~\ref{thm:tight-zero-rate}:
\begin{enumerate}[label=\textup{(\roman*)}]
  \item if \(P_O\) is uniform on \(S_O\), then
        \begin{equation}\label{eq:uniform-zero-rate}
          R_{\mathrm{sem}}(0;\,d_{\Cn},\,P_O)
          =
          \frac{k}{|S_O|}\log k;
        \end{equation}
  \item if, in addition, \(S_O\subseteq\hat S_O\),
        \(P_O\) has full support, and \(J\neq\varnothing\), then
        \begin{equation}\label{eq:zero-rate-ratio}
          R_{\mathrm{sem}}(0;\,d_{\Cn},\,P_O)
          <
          R(0;\,d_H,\,P_O)
          =
          H(P_O).
        \end{equation}
\end{enumerate}
\end{corollary}

\begin{proof}
For~\textup{(i)}, under the uniform source,
\(P_A=k/|S_O|\) and \(\pi_A\) is uniform on \(A\), so
\(H(\pi_A)=\log k\).
Substituting into~\eqref{eq:tight-R0} gives
\eqref{eq:uniform-zero-rate}.

For~\textup{(ii)}, write \(\pi_J\) for the source distribution
conditioned on \(J\).
Then
\[
  H(P_O)
  =
  h_b(P_A)+P_A H(\pi_A)+P_J H(\pi_J).
\]
If \(J\neq\varnothing\) and \(P_O\) has full support, then
\(0<P_A<1\), so \(h_b(P_A)>0\).
Hence
\[
  H(P_O)>P_A H(\pi_A).
\]
Using Theorem~\ref{thm:tight-zero-rate} and the classical identity
\(R(0;\,d_H,\,P_O)=H(P_O)\) for zero Hamming distortion on a
reconstruction alphabet containing \(S_O\)
\cite[Ch.~5]{cover2006elements}\cite[Ch.~7]{csiszar2011information},
which is valid here because \(S_O\subseteq\hat S_O\), completes the
proof.
\end{proof}

\begin{remark}[Interpretation]
\label{rem:why-proof-system}
Theorem~\ref{thm:tight-zero-rate} is the point at which deductive
structure enters the information theory in a decisive way.
Under Hamming fidelity, every source symbol must remain
distinguishable and the zero-distortion rate is \(H(P_O)\).
Under closure fidelity, only the irredundant core must remain
distinguishable; the redundant part is absorbed into a single
zero-cost component.
This is why the exact answer is \(P_A H(\pi_A)\) rather than
\(H(P_O)\).
\end{remark}

\begin{remark}[Beyond disjointness]
\label{rem:general-confusable-case}
Assumption~\ref{assump:core-disjoint} identifies a transparent regime in which
the zero-distortion converse collapses to a single-letter Shannon-entropy
expression. When disjointness fails, overlaps among the sets
\(\{R_{\Cn}(a):a\in A\}\) create a genuine core-confusability phenomenon.
Section~\ref{sec:confusable} resolves the resulting \(D=0\) limit exactly via a
hypergraph-entropy characterization, with a reduction to graph entropy under a
pairwise realisability condition.

What remains challenging is primarily algorithmic: computing or approximating
the induced hypergraph/graph entropy efficiently, and extending the same
combinatorial characterisation beyond \(D=0\) under closure-based (non-binary)
distortions.
\end{remark}

\section{Zero-Distortion Rate with Confusable Cores}
\label{sec:confusable}

Section~\ref{sec:exact-zero-rate} gave an exact single-letter formula for
\(R_{\mathrm{sem}}(0)\) under the disjointness (nonconfusability) hypothesis.
We now drop Assumption~\ref{assump:core-disjoint}.
When zero-distortion reconstruction sets overlap on the core, a single
reconstruction symbol may serve multiple core symbols simultaneously, and
the optimal zero-distortion rate is no longer governed by Shannon entropy
alone. The correct object is the family of core subsets that admit a
\emph{common} zero-distortion reconstruction, which naturally forms a
hypergraph; its associated hypergraph entropy yields the exact
zero-distortion rate \cite{korner1971coding,alon1996source,orlitsky2001coding}.

Throughout this section, let \(A=\Atom(S_O)\), \(P_A=P_O(A)\), and when \(P_A>0\)
let \(\pi_A(\cdot)=P_O(\cdot\mid A)\) as in Section~\ref{sec:exact-zero-rate}.

\subsection{Zero-distortion reconstruction sets and the core hypergraph}
\label{subsec:confusable-hypergraph}

Recall the zero-distortion reconstruction sets under closure distortion:
\[
  R_{\Cn}(s)
  :=
  \{\hat s\in\hat S_O:\ d_{\Cn}(s,\hat s\mid S_O)=0\},
  \qquad s\in S_O.
\]
If some positive-probability source symbol has an empty zero-distortion
reconstruction set, the zero-distortion constraint is infeasible; see
Lemma~\ref{lem:zero-distortion-infeasible-infty}.
We therefore isolate the feasible regime.

\begin{lemma}[Zero distortion implies closure membership]
\label{lem:zero-distortion-implies-in-closure}
If \(d_{\Cn}(s,\hat s\mid S_O)=0\), then \(\hat s\in \Cn(S_O)\).
\end{lemma}

\begin{proof}
If \(d_{\Cn}(s,\hat s\mid S_O)=0\), then
\(\Cn((S_O\setminus\{s\})\cup\{\hat s\})=\Cn(S_O)\).
By reflexivity, \(\hat s\in \Cn((S_O\setminus\{s\})\cup\{\hat s\})\), hence \(\hat s\in \Cn(S_O)\).
\end{proof}

\begin{lemma}[Zero-distortion infeasibility implies infinite rate]
\label{lem:zero-distortion-infeasible-infty}
If there exists \(s\in S_O\) with \(P_O(s)>0\) and \(R_{\Cn}(s)=\varnothing\),
then \(R_{\mathrm{sem}}(0;\,d_{\Cn},\,P_O)=+\infty\).
\end{lemma}

\begin{proof}
If \(R_{\Cn}(s)=\varnothing\), then \(d_{\Cn}(s,\hat s\mid S_O)>0\) for all \(\hat s\in\hat S_O\).
Thus every test channel incurs strictly positive expected distortion whenever \(P_O(s)>0\),
so the zero-distortion constraint is infeasible and the minimum mutual information is \(+\infty\).
\end{proof}

\begin{assumption}[Core coverage at zero distortion]
\label{assump:core-coverage}
Either \(P_A=0\) and \(\hat S_O\cap\Cn(S_O)\neq\varnothing\), or \(P_A>0\) and
\(R_{\Cn}(a)\neq\varnothing\) for every \(a\in A\).
\end{assumption}

Assume \(A\neq\varnothing\). Define the \emph{core zero-distortion hypergraph}
\[
  \Gamma_0
  :=
  \Bigl\{
    W\subseteq A:\ W\neq\emptyset,\ 
    \bigcap_{a\in W}R_{\Cn}(a)\neq\emptyset
  \Bigr\}.
\]
Equivalently, \(W\in\Gamma_0\) iff there exists a single reconstruction symbol
\(\hat s\in\hat S_O\) that is zero-distortion correct for every \(a\in W\).
Note that \(\Gamma_0\) is downward closed: if \(W\in\Gamma_0\) and
\(\emptyset\neq W'\subseteq W\), then \(\cap_{a\in W'}R_{\Cn}(a)\supseteq\cap_{a\in W}R_{\Cn}(a)\neq\emptyset\).

\subsection{Hypergraph-entropy characterization of \(R_{\mathrm{sem}}(0)\)}
\label{subsec:confusable-exact}

For \(P_A>0\), let \(\mathsf{A}_o\sim \pi_A\). Consider a random nonempty subset
\(W\subseteq A\) taking values in \(\Gamma_0\) and satisfying the constraint
\(\mathsf{A}_o\in W\) almost surely. Define the induced hypergraph entropy
\begin{equation}\label{eq:hypergraph-entropy-Gamma0}
  H_{\Gamma_0}(\pi_A)
  :=
  \min_{\substack{P_{W\mid A}:\\
  \mathsf{A}_o\in W,\ W\in\Gamma_0\ \text{a.s.}}}
  I(\mathsf{A}_o;W).
\end{equation}

\begin{theorem}[General confusable case: exact zero-distortion rate]
\label{thm:confusable-hypergraph-exact}
Assume \(A\neq\varnothing\) and Assumption~\ref{assump:core-coverage}.
Then
\begin{equation}\label{eq:R0-hypergraph}
  R_{\mathrm{sem}}(0;\,d_{\Cn},\,P_O)
  =
  P_A\,H_{\Gamma_0}(\pi_A),
\end{equation}
with the convention that the right-hand side is \(0\) when \(P_A=0\).
\end{theorem}

\begin{proof}[Proof sketch]
Assumption~\ref{assump:core-coverage} rules out the infeasible case of
Lemma~\ref{lem:zero-distortion-infeasible-infty} and ensures that zero distortion is achievable.
The converse constructs, from any zero-distortion decoder output, a random set \(W\in\Gamma_0\)
containing the true core symbol and then applies data processing to lower bound the rate by
\(P_A H_{\Gamma_0}(\pi_A)\).
Achievability uses a selector \(\psi(W)\in\cap_{a\in W}R_{\Cn}(a)\) and a random-set kernel
\(P_{W\mid A}\) to build a zero-distortion test channel whose rate matches this bound; the
redundant part is handled using Proposition~\ref{prop:redundant-free} together with
Lemma~\ref{lem:zero-distortion-implies-in-closure}.
Full details are given in Appendix~\ref{app:proof-confusable-hypergraph-exact}.
\end{proof}

\subsection{Graph-entropy reduction under pairwise realisability}
\label{subsec:confusable-graph}

Define the \emph{core incompatibility graph} \(G_{\mathrm{inc}}\) on vertex set
\(A\) by
\[
  \{a,a'\}\in E(G_{\mathrm{inc}})
  \iff
  R_{\Cn}(a)\cap R_{\Cn}(a')=\emptyset.
\]
Thus two core symbols are adjacent iff no single reconstruction symbol can be
zero-distortion correct for both.

\begin{assumption}[Pairwise realisability]
\label{assump:pairwise-realisability}
For every nonempty \(W\subseteq A\),
\[
  \Bigl[
    \forall\,a\neq a'\in W,\ 
    R_{\Cn}(a)\cap R_{\Cn}(a')\neq\emptyset
  \Bigr]
  \Longrightarrow
  \bigcap_{a\in W}R_{\Cn}(a)\neq\emptyset.
\]
\end{assumption}

Define K\"orner's graph entropy (random independent-set formulation) as
\begin{equation}\label{eq:graph-entropy}
  H_{G_{\mathrm{inc}}}(\pi_A)
  :=
  \min_{\substack{P_{W\mid A}:\\
  \mathsf{A}_o\in W,\ W\in\Ind(G_{\mathrm{inc}})\ \text{a.s.}}}
  I(\mathsf{A}_o;W),
\end{equation}
where \(\Ind(G_{\mathrm{inc}})\) denotes the family of nonempty independent sets.

\begin{corollary}[Graph-entropy law under pairwise realisability]
\label{cor:confusable-graph-entropy}
Assume \(A\neq\varnothing\), Assumption~\ref{assump:core-coverage}, and
Assumption~\ref{assump:pairwise-realisability}. Then
\[
  R_{\mathrm{sem}}(0;\,d_{\Cn},\,P_O)
  =
  P_A\,H_{G_{\mathrm{inc}}}(\pi_A).
\]
\end{corollary}

\begin{proof}[Proof sketch]
Under pairwise realisability, a subset \(W\subseteq A\) has a common
zero-distortion reconstruction iff it contains no incompatible pair, i.e.,
iff \(W\in\Ind(G_{\mathrm{inc}})\). Thus \(\Gamma_0=\Ind(G_{\mathrm{inc}})\), so
\(H_{\Gamma_0}(\pi_A)=H_{G_{\mathrm{inc}}}(\pi_A)\). Apply
Theorem~\ref{thm:confusable-hypergraph-exact}.
\end{proof}

\subsection{Consistency with the disjoint case}
If Assumption~\ref{assump:core-disjoint} holds, then \(G_{\mathrm{inc}}\) is
complete and \(\Ind(G_{\mathrm{inc}})\) consists only of singletons. Therefore
\(H_{G_{\mathrm{inc}}}(\pi_A)=H(\pi_A)\), and
\eqref{eq:R0-hypergraph} reduces to Theorem~\ref{thm:tight-zero-rate}.

\section{Core Decomposition of the Full Rate--Distortion Function}
\label{sec:core-decomposition}

Section~\ref{sec:exact-zero-rate} identified the exact
zero-distortion rate under closure fidelity.
We now extend the analysis to arbitrary distortion levels.
The main point is that the redundant part \(J=S_O\setminus A\)
remains invisible not only at \(D=0\), but across the entire
rate--distortion function.
Thus the full problem decomposes into a classical
rate--distortion problem on the irredundant core, scaled by the
source mass of that core.

\begin{definition}[Core sub-source rate--distortion function]
\label{def:core-rd}
Let \(A=\Atom(S_O)\) and let \(\pi_A\) be the source
distribution conditioned on \(A\).
Assume \(\hat S_O\subseteq\Cn(S_O)\cap\mathbb{S}_O\).
The \emph{core sub-source rate--distortion function} is
\begin{equation}\label{eq:core-rd}
  R^{(A)}(D';\,d_{\Cn},\,\pi_A)
  :=
  \min_{\substack{P_{\hat S\mid A}:\\
  \E[d_{\Cn}(\mathsf{A}_o,\hat {\mathsf{S}}_o)]\le D'}}
  I(\mathsf{A}_o;\hat{\mathsf{S}}_o),
\end{equation}
where \(\mathsf{A}_o\sim\pi_A\), the reconstruction alphabet is
\(\hat S_O\), and the distortion is
\(d_{\Cn}(\cdot,\cdot\mid S_O)\) restricted to inputs in \(A\).
\end{definition}

\begin{theorem}[Core decomposition of the full
rate--distortion function]
\label{thm:rd-decomposition}
Assume \(\hat S_O\subseteq\Cn(S_O)\cap\mathbb{S}_O\).
Then, for every \(D\ge 0\),
\begin{equation}\label{eq:rd-decomp}
  R_{\mathrm{sem}}(D;\,d_{\Cn},\,P_O)
  \;=\;
  P_A\cdot
  R^{(A)}\!\Bigl(\frac{D}{P_A};\,d_{\Cn},\,\pi_A\Bigr),
\end{equation}
with the convention that the right-hand side is \(0\) when
\(P_A=0\).

In particular:
\begin{enumerate}[label=\textup{(\roman*)}]
  \item under the additional hypotheses
        \(A\subseteq\hat S_O\) and
        Assumption~\textup{\ref{assump:core-disjoint}},
        one has
        \(R_{\mathrm{sem}}(0;\,d_{\Cn},\,P_O)=P_AH(\pi_A)\),
        recovering Theorem~\textup{\ref{thm:tight-zero-rate}};
  \item \(R_{\mathrm{sem}}(\cdot)\) is convex,
      non-increasing, and reducible to a finite-alphabet
      core problem that is amenable to standard numerical
      evaluation;
  \item redundant states are invisible to both rate and
        distortion.
\end{enumerate}
\end{theorem}

\begin{proof}[Proof sketch]
By Proposition~\ref{prop:redundant-free}, every redundant state contributes
zero closure distortion against every reconstruction in
\(\hat S_O\subseteq\Cn(S_O)\), so both the distortion constraint and the
mutual-information minimization reduce to the irredundant core.
A lower bound follows by introducing the auxiliary variable
\(T(s_o)=s_o\) on \(A\) and \(T(s_o)=*\) on \(J\), and then applying
entropy concavity exactly as in the converse of
Theorem~\ref{thm:tight-zero-rate}.
Achievability is obtained by extending an optimal test channel for the
core sub-source to the redundant part using a common output law.
The full proof is given in Appendix~\ref{app:proof-rd-decomposition}.
\end{proof}

\begin{remark}[Disjointness is needed only for the exact
zero-distortion formula]
\label{rem:rd-decomp-vs-disjoint}
The decomposition theorem itself does \emph{not} require
Assumption~\ref{assump:core-disjoint}.
That assumption is needed only to identify the
core-only zero-distortion term in closed form as
\(H(\pi_A)\).
Thus the structural reduction to the core is more general than
the exact zero-distortion evaluation.
\end{remark}

\begin{remark}[What remains nonclassical at positive distortion]
\label{rem:positive-D-nonclassical}
Theorem~\ref{thm:rd-decomposition} reduces the full source to its
core, but the reduced problem is still nonclassical because the
distortion on the core is the closure distortion
\(d_{\Cn}(\cdot,\cdot\mid S_O)\), not a symbolwise distortion.
Thus the theorem isolates the precise source of the semantic
gain: redundancy disappears entirely, while the nontrivial
positive-distortion geometry is concentrated on the core.
\end{remark}

\section{Limited Inference and the Rate--Depth--Distortion Tradeoff}
\label{sec:limited-inference}

Sections~\ref{sec:exact-zero-rate} and~\ref{sec:core-decomposition}
treat the decoder as having unrestricted deductive power.
We now introduce a bounded inference-depth budget.
The resulting fidelity criterion distinguishes statements that are
re-derivable \emph{eventually} from those re-derivable
\emph{within \(\delta\) iterations of \(T_{\mathsf{PS}}\)}.
This yields a natural rate--depth--distortion tradeoff.

\subsection{Finite-Step Deduction and Derivation Depth}
\label{subsec:limited-inference-model}

\begin{assumption}[Finite-step closure dynamics]
\label{assump:finite-step-closure}
For the proof systems considered in this section, there exists a
computable monotone operator \(T_{\mathsf{PS}}\) on finite subsets
of \(\mathbb{S}_O\) such that, for every finite
\(B\subseteq\mathbb{S}_O\),
\begin{enumerate}[label=\textup{(T\arabic*)}]
  \item \(T_{\mathsf{PS}}(B)\supseteq B\);
  \item \(B\subseteq B' \Rightarrow
        T_{\mathsf{PS}}(B)\subseteq T_{\mathsf{PS}}(B')\);
  \item with \(T_{\mathsf{PS}}^0(B):=B\) and
        \(T_{\mathsf{PS}}^{n+1}(B):=
        T_{\mathsf{PS}}(T_{\mathsf{PS}}^n(B))\),
        one has
        \[
          \Cn(B)=\bigcup_{n\ge 0}T_{\mathsf{PS}}^n(B);
        \]
  \item the sequence \(T_{\mathsf{PS}}^n(B)\) stabilizes in
        finitely many iterations.
\end{enumerate}
This covers standard Datalog and Horn-clause settings over finite
active domains
\cite{van1976semantics,ceri1989you,abiteboul1995foundations,
dantsin2001complexity}.
\end{assumption}

\begin{definition}[Derivation depth]
\label{def:derivation-depth}
For finite \(B\subseteq\mathbb{S}_O\) and \(s\in\mathbb{S}_O\),
define
\[
  \Dd(s\mid B)
  :=
  \min\{n\ge 0:\ s\in T_{\mathsf{PS}}^n(B)\},
\]
with \(\Dd(s\mid B):=\infty\) if \(s\notin\Cn(B)\).
\end{definition}

\begin{definition}[Maximum intrinsic derivation depth]
\label{def:max-depth}
Let \(A=\Atom(S_O)\).
The \emph{maximum intrinsic derivation depth} of \(S_O\) is
\[
  \mathsf{D_d}
  :=
  \max_{q\in S_O}\Dd(q\mid A),
\]
with the convention \(\max\varnothing:=0\).
\end{definition}

\begin{assumption}[Order-robust core for the bounded-inference section]
\label{assump:atom-equals-essential}
Throughout Section~\ref{sec:limited-inference}, we additionally assume
\[
  A=\Atom(S_O)=\Ess(S_O)
  :=
  \{s\in S_O:\ s\notin \Cn(S_O\setminus\{s\})\}.
\]
Equivalently, every element of the chosen core is globally
nonderivable from the remaining stored source symbols.
This assumption is satisfied, for example, under the
essential-generation condition of
Proposition~\ref{prop:essential-order-invariance}.
\end{assumption}

\begin{remark}[Why the extra assumption is needed here]
\label{rem:why-extra-assumption-limited-inference}
The family \(A_\delta\) is defined directly from derivability out of
\(S_O\setminus\{s\}\), so as \(\delta\) increases it converges to the
order-free essential set \(\Ess(S_O)\), not in general to an
arbitrary order-dependent deletion core \(\Atom(S_O)\).
Assumption~\ref{assump:atom-equals-essential} is therefore exactly
what ensures that the large-\(\delta\) endpoint of the
bounded-inference theory coincides with the zero-distortion law of
Section~\ref{sec:exact-zero-rate}.
\end{remark}

\subsection{Depth-Constrained Fidelity}
\label{subsec:delta-distortion}

\begin{definition}[\(\delta\)-step closure distortion]
\label{def:delta-closure-distortion}
For integer \(\delta\ge 0\), define
\begin{equation}\label{eq:d-Cn-delta}
  d_{\Cn}^{\delta}(s_o,\hat s\mid S_O)
  :=
  \begin{cases}
    0, &
    \text{if } S_O\subseteq
    T_{\mathsf{PS}}^\delta\bigl((S_O\setminus\{s_o\})\cup
    \{\hat s\}\bigr),\\[4pt]
    1, & \text{otherwise}.
  \end{cases}
\end{equation}
\end{definition}

\begin{remark}[Interpretation of \(d_{\Cn}^{\delta}\)]
\label{rem:dCn-delta-interpretation}
The distortion \(d_{\Cn}^{\delta}\) is a bounded-inference
fidelity criterion.
It asks whether the entire source set \(S_O\) can be
reconstructed within at most \(\delta\) iterations of the
immediate-consequence operator after replacing \(s_o\) by
\(\hat s\).

At \(\delta=0\), one has
\[
  d_{\Cn}^{0}(s_o,\hat s\mid S_O)=\mathbf{1}[s_o\neq\hat s]
  =d_H(s_o,\hat s).
\]
Thus the zero-inference endpoint recovers classical symbolwise
fidelity.
\end{remark}

\begin{definition}[\(\delta\)-redundancy and
\(\delta\)-irredundant core]
\label{def:delta-core}
A state \(s\in S_O\) is \emph{\(\delta\)-redundant} if
\[
  s\in T_{\mathsf{PS}}^\delta(S_O\setminus\{s\}).
\]
The corresponding \emph{\(\delta\)-irredundant core} is
\[
  A_\delta
  :=
  \Atom_\delta(S_O)
  :=
  \{s\in S_O:\ s\notin
  T_{\mathsf{PS}}^\delta(S_O\setminus\{s\})\}.
\]
Write
\[
  P_\delta:=P_O(A_\delta),
  \qquad
  \pi_\delta(s):=\frac{P_O(s)}{P_\delta},
  \quad s\in A_\delta,
\]
whenever \(P_\delta>0\).
\end{definition}

\begin{lemma}[\(\delta\)-redundant states are distortion-free]
\label{lem:delta-free-substitution}
If \(s\notin A_\delta\), then
\[
  d_{\Cn}^{\delta}(s,\hat s\mid S_O)=0,
  \qquad \forall\,\hat s\in\hat S_O.
\]
\end{lemma}

\begin{proof}
If \(s\notin A_\delta\), then by definition
\[
  s\in T_{\mathsf{PS}}^\delta(S_O\setminus\{s\}).
\]
By monotonicity of \(T_{\mathsf{PS}}^\delta\),
\[
  T_{\mathsf{PS}}^\delta(S_O\setminus\{s\})
  \subseteq
  T_{\mathsf{PS}}^\delta\bigl((S_O\setminus\{s\})\cup\{\hat s\}\bigr).
\]
Hence both \(S_O\setminus\{s\}\) and \(s\) belong to the latter
set, so the whole source set \(S_O\) is recovered within
\(\delta\) iterations.
Therefore \(d_{\Cn}^{\delta}(s,\hat s\mid S_O)=0\).
\end{proof}

\begin{proposition}[\(\delta\)-core filtration]
\label{prop:filtration}
The family \(\{A_\delta\}_{\delta\ge 0}\) satisfies:
\begin{enumerate}[label=\textup{(\roman*)}]
  \item \(A_0=S_O\);
  \item \(A_{\delta+1}\subseteq A_\delta\) for every
        \(\delta\ge 0\);
  \item under Assumption~\ref{assump:atom-equals-essential},
        \(A_\delta=A\) for every \(\delta\ge \mathsf{D_d}\).
\end{enumerate}
\end{proposition}

\begin{proof}
Part~\textup{(i)} follows from
\(T_{\mathsf{PS}}^0(S_O\setminus\{s\})=S_O\setminus\{s\}\), which
never contains \(s\).

For~\textup{(ii)}, monotonicity in the number of iterations gives
\[
  T_{\mathsf{PS}}^\delta(S_O\setminus\{s\})
  \subseteq
  T_{\mathsf{PS}}^{\delta+1}(S_O\setminus\{s\}),
\]
so once \(s\) becomes \(\delta\)-redundant, it remains
\((\delta+1)\)-redundant.

For~\textup{(iii)}, let \(j\in J=S_O\setminus A\).
Since \(\Cn(A)=\Cn(S_O)\), one has \(j\in\Cn(A)\), and by the
definition of \(\mathsf{D_d}\),
\[
  j\in T_{\mathsf{PS}}^{\mathsf{D_d}}(A).
\]
Because \(A\subseteq S_O\setminus\{j\}\), monotonicity yields
\[
  j\in T_{\mathsf{PS}}^{\mathsf{D_d}}(S_O\setminus\{j\}),
\]
hence \(j\notin A_\delta\) for every \(\delta\ge \mathsf{D_d}\).

Conversely, if \(a\in A\), then by
Assumption~\ref{assump:atom-equals-essential},
\[
  a\notin \Cn(S_O\setminus\{a\}).
\]
Therefore \(a\notin T_{\mathsf{PS}}^\delta(S_O\setminus\{a\})\) for
every finite \(\delta\), because
\(T_{\mathsf{PS}}^\delta(S_O\setminus\{a\})
 \subseteq \Cn(S_O\setminus\{a\})\).
Hence \(a\in A_\delta\) for all \(\delta\), proving
\(A_\delta=A\) when \(\delta\ge \mathsf{D_d}\).
\end{proof}

\begin{assumption}[Disjoint zero-distortion sets for the
\(\delta\)-core]
\label{assump:delta-core-disjoint}
For each \(\delta\ge 0\) and each \(a\in A_\delta\), define
\[
  R_\delta(a)
  :=
  \{\hat s\in\hat S_O:\ d_{\Cn}^{\delta}(a,\hat s\mid S_O)=0\}.
\]
We assume that for distinct \(a_1,a_2\in A_\delta\),
\[
  R_\delta(a_1)\cap R_\delta(a_2)=\varnothing.
\]
\end{assumption}

\subsection{Exact Rate--Depth--Distortion Laws}
\label{subsec:rate-depth-laws}

\begin{definition}[Rate--depth--distortion function]
\label{def:rdd-function}
For \(\delta\ge 0\), define
\begin{equation}\label{eq:rdd-def}
  R_{\mathrm{sem}}(D,\delta;\,d_{\Cn}^{\delta},\,P_O)
  :=
  \min_{\substack{P_{\hat S\mid S}:\\
  \E[d_{\Cn}^{\delta}({\mathsf{S}}_o,\hat {\mathsf{S}}_o)]\le D}}
  I(\mathsf{S}_o;\hat{\mathsf{S}}_o).
\end{equation}
\end{definition}

\begin{theorem}[Exact zero-distortion rate--depth tradeoff]
\label{thm:rate-depth}
Fix \(\delta\ge 0\).
Assume \(A_\delta\subseteq\hat S_O\) and that
Assumption~\ref{assump:delta-core-disjoint} holds at this \(\delta\).
Then
\begin{equation}\label{eq:rate-depth}
  R_{\mathrm{sem}}(0,\delta;\,d_{\Cn}^{\delta},\,P_O)
  \;=\;
  P_\delta\,H(\pi_\delta),
\end{equation}
with the convention that the right-hand side is \(0\) when
\(P_\delta=0\).
\end{theorem}

\begin{proof}[Proof sketch]
The proof is the exact analogue of
Theorem~\ref{thm:tight-zero-rate}, with
\(A\), \(P_A\), and \(\pi_A\) replaced by
\(A_\delta\), \(P_\delta\), and \(\pi_\delta\), and with
Lemma~\ref{lem:delta-free-substitution} playing the role of
Proposition~\ref{prop:redundant-free}.

Indeed, define
\[
  T_\delta(s_o)
  :=
  \begin{cases}
    s_o, & s_o\in A_\delta,\\
    *,   & s_o\notin A_\delta.
  \end{cases}
\]
Under zero \(\delta\)-distortion, every \(a\in A_\delta\) must be
reconstructed inside its zero-distortion set \(R_\delta(a)\), and
these sets are pairwise disjoint by
Assumption~\ref{assump:delta-core-disjoint}.
For \(s_o\notin A_\delta\), the output is unconstrained by
Lemma~\ref{lem:delta-free-substitution}.
The same disjoint-support entropy identity as in the proof of
Theorem~\ref{thm:tight-zero-rate} therefore yields the converse
\[
  I({\mathsf{S}}_o;\hat {\mathsf{S}}_o)\ge P_\delta H(\pi_\delta).
\]
Achievability is obtained by reproducing each \(a\in A_\delta\)
exactly and mapping every \(s_o\notin A_\delta\) to a common
distribution \(\pi_\delta\) supported on \(A_\delta\).
Then zero \(\delta\)-distortion holds, and the resulting mutual
information is exactly \(P_\delta H(\pi_\delta)\).
\end{proof}

\begin{corollary}[Endpoints and monotonicity of the zero-distortion
rate--depth function]
\label{cor:rate-depth-endpoints}
Assume \(S_O\subseteq\hat S_O\), that
Assumption~\ref{assump:delta-core-disjoint} holds for every
\(\delta\ge 0\), and that
Assumption~\ref{assump:atom-equals-essential} holds.
Then:
\begin{enumerate}[label=\textup{(\roman*)}]
  \item \(R_{\mathrm{sem}}(0,0;\,d_{\Cn}^{0},\,P_O)=H(P_O)\);
  \item \(R_{\mathrm{sem}}(0,\delta;\,d_{\Cn}^{\delta},\,P_O)
        =P_AH(\pi_A)\) for all \(\delta\ge \mathsf{D_d}\);
  \item the function
        \(\delta\mapsto
        R_{\mathrm{sem}}(0,\delta;\,d_{\Cn}^{\delta},\,P_O)\)
        is non-increasing.
\end{enumerate}
\end{corollary}

\begin{proof}
Part~\textup{(i)} follows from Theorem~\ref{thm:rate-depth} at
\(\delta=0\). Indeed, \(A_0=S_O\), so \(P_0=1\) and
\(\pi_0=P_O\), which gives
\[
  R_{\mathrm{sem}}(0,0;\,d_{\Cn}^{0},\,P_O)=H(P_O).
\]

Part~\textup{(ii)} follows from
Proposition~\ref{prop:filtration}\textup{(iii)} and
Theorem~\ref{thm:rate-depth}: for every \(\delta\ge \mathsf{D_d}\),
one has \(A_\delta=A\), hence \(P_\delta=P_A\) and
\(\pi_\delta=\pi_A\).

For part~\textup{(iii)}, note that if \(\delta_1\le \delta_2\), then
\[
  d_{\Cn}^{\delta_2}(s_o,\hat s\mid S_O)
  \le
  d_{\Cn}^{\delta_1}(s_o,\hat s\mid S_O)
\]
for all \((s_o,\hat s)\), because allowing more iterations (of \(T_{\mathsf{PS}}\)) can
only enlarge the zero-distortion feasible set. Therefore the feasible
test-channel set for zero \(\delta_2\)-distortion contains that for
zero \(\delta_1\)-distortion, and thus
\[
  R_{\mathrm{sem}}(0,\delta_2;\,d_{\Cn}^{\delta_2},\,P_O)
  \le
  R_{\mathrm{sem}}(0,\delta_1;\,d_{\Cn}^{\delta_1},\,P_O).
\]
So \(\delta\mapsto R_{\mathrm{sem}}(0,\delta)\) is non-increasing.
\end{proof}

\begin{definition}[\(\delta\)-core sub-source rate--distortion function]
\label{def:delta-core-rd}
For \(\delta\ge 0\) and \(P_\delta>0\), define
\begin{equation}\label{eq:delta-core-rd}
  R^{(A_\delta)}(D';\,d_{\Cn}^{\delta},\,\pi_\delta)
  :=
  \min_{\substack{P_{\hat S\mid A_\delta}:\\
  \E[d_{\Cn}^{\delta}(A_{\delta,o},\hat {\mathsf{S}}_o)]\le D'}}
  I(\mathsf{A}_{\delta,o};\hat{\mathsf{S}}_o),
\end{equation}
where \(\mathsf{A}_{\delta,o}\sim\pi_\delta\), the reconstruction
alphabet is \(\hat S_O\), and the distortion is
\(d_{\Cn}^{\delta}(\cdot,\cdot\mid S_O)\) restricted to inputs in
\(A_\delta\).
\end{definition}

\begin{theorem}[Full rate--depth--distortion decomposition]
\label{thm:rate-depth-distortion}
For every \(\delta\ge 0\),
\begin{equation}\label{eq:rdd-surface}
  R_{\mathrm{sem}}(D,\delta;\,d_{\Cn}^{\delta},\,P_O)
  \;=\;
  P_\delta\cdot
  R^{(A_\delta)}\!\Bigl(\frac{D}{P_\delta};\,
  d_{\Cn}^{\delta},\,\pi_\delta\Bigr),
\end{equation}
with the convention that the right-hand side is \(0\) when
\(P_\delta=0\), where \(R^{(A_\delta)}\) is the
rate--distortion function of the \(\delta\)-core sub-source
\((A_\delta,\pi_\delta)\) under the distortion
\(d_{\Cn}^{\delta}(\cdot,\cdot\mid S_O)\) restricted to inputs in
\(A_\delta\).

In particular, at \(\delta=0\), since \(A_0=S_O\) and
\(d_{\Cn}^{0}=d_H\), one recovers the ordinary Hamming
rate--distortion problem on source alphabet \(S_O\) and
reconstruction alphabet \(\hat S_O\):
\[
  R_{\mathrm{sem}}(D,0;\,d_{\Cn}^{0},\,P_O)
  =
  \min_{\substack{P_{\hat S\mid S}:\\
  \E[d_H({\mathsf{S}}_o,\hat {\mathsf{S}}_o)]\le D}}
  I(\mathsf{S}_o;\hat{\mathsf{S}}_o).
\]
\end{theorem}

\begin{proof}[Proof sketch]
By Lemma~\ref{lem:delta-free-substitution}, every symbol outside
\(A_\delta\) is distortion-free under \(d_{\Cn}^{\delta}\), so both
the distortion constraint and the mutual-information minimization
reduce to the \(\delta\)-core.
The lower bound follows by introducing
\(T_\delta(s_o)=s_o\) on \(A_\delta\) and \(T_\delta(s_o)=*\) on its
complement, exactly as in the proof of
Theorem~\ref{thm:rd-decomposition}.
Achievability is obtained by extending an optimal test channel for the
\(\delta\)-core sub-source to the \(\delta\)-redundant part using a
common output law.
At \(\delta=0\), one has \(A_0=S_O\) and \(d_{\Cn}^{0}=d_H\), so the
formula reduces to the classical Hamming rate--distortion function
\cite{shannon1959coding,cover2006elements,csiszar2011information}.
Full details are given in
Appendix~\ref{app:proof-rate-depth-distortion}.
\end{proof}

\begin{remark}[Interpretation]
\label{rem:depth-interpretation}
The family \(\{A_\delta\}_{\delta\ge 0}\) provides a quantitative
depth-for-rate exchange.
As \(\delta\) increases, more states become reconstructible within
the allowed depth budget and therefore disappear from the
effective source.
At \(\delta=0\), nothing is free and the theory collapses to
classical source coding.
At \(\delta\ge \mathsf{D_d}\), all eventually derivable shortcuts
become free and, under the standing assumptions guaranteeing the
large-\(\delta\) endpoint identification, one recovers the deductive
compression law of Section~\ref{sec:exact-zero-rate}.
\end{remark}

\begin{remark}[Endpoint caution]
\label{rem:depth-endpoint-caution}
The bounded-inference distortion \(d_{\Cn}^{\delta}\) is a
binary feasibility distortion, whereas the unconstrained closure
distortion \(d_{\Cn}\) of Sections~\ref{sec:deductive-source} and
\ref{sec:core-decomposition} is a Jaccard-type distortion.
Thus Theorem~\ref{thm:rate-depth-distortion} should be interpreted
as the exact bounded-inference analogue of the
rate--distortion function.
Its zero-distortion endpoint agrees with
Theorem~\ref{thm:tight-zero-rate}, but for \(D>0\) it is a distinct
operational distortion family.
\end{remark}

\section{Consequences: Separation, Converse Bounds, and One Heterogeneous Extension}
\label{sec:consequences}

The preceding sections are source-coding results.
We now record three consequences.
First, because the relevant distortions are single-letter and
bounded, the classical source--channel separation theorem applies
verbatim
\cite{shannon1959coding,cover2006elements,csiszar2011information}.
Second, the core/redundant decomposition yields sharper converse
bounds than symbolwise fidelity.
Third, the same zero-distortion mechanism extends to a restricted
reconstruction alphabet, giving a compact heterogeneous-receiver
extension.

\begin{remark}[Degenerate core cases]
\label{rem:degenerate-core-cases}
If \(A=\varnothing\) and
\(\hat S_O\cap\Cn(S_O)\neq\varnothing\), then a constant
reconstruction to any
\(\hat s^\star\in \hat S_O\cap\Cn(S_O)\) achieves zero closure
distortion, so the effective zero-distortion rate is \(0\).
If \(|A|=1\), then all entropy and blocklength expressions involving
the core reduce to their obvious one-symbol values, and any term of
the form \(\epsilon\log(|A|-1)\) should be read through the trivial
one-symbol case rather than literally as a logarithm.
Accordingly, the nontrivial converse and operational bounds below are
primarily of interest when \(|A|\ge 2\).
\end{remark}

\subsection{Separation and Operational Blocklength}
\label{subsec:separation}

\begin{remark}[Two operational regimes]
\label{rem:two-operational-regimes}
The results in this section use two standard operational viewpoints.
In the \emph{source-coding regime}, one transmits an i.i.d.\ block of
statement symbols drawn from the fixed deductive source
\((S_O,P_O)\), leading to rate--distortion and separation statements.
In the \emph{message-set regime}, one treats elements of a finite set
\(\mathcal M\subseteq S_O\) as equiprobable messages and studies the
minimum blocklength needed for reliable transmission.
The former is the asymptotic rate--distortion viewpoint; the latter is
the finite-message operational viewpoint underlying the converse and
achievability bounds below.
\end{remark}

\begin{definition}[Block code and reliability criteria]
\label{def:semantic-codebook}
Let \(W:\mathcal X\rightsquigarrow\mathcal Y\) be a discrete
memoryless channel with Shannon capacity \(C(W)\).
An \emph{\((n,M)\) block code} for message set
\(\mathcal M\subseteq S_O\) consists of an encoder
\(f_n:\mathcal M\to\mathcal X^n\) and a decoder
\(g_n:\mathcal Y^n\to\hat S_O\), where \(|\mathcal M|=M\).

When message \(m\in\mathcal M\) is sent, write
\(\hat{\mathsf S}_o^{(m)}:=g_n(\hat{\mathsf Y}^n)\) for the
decoder output.
The \emph{Hamming error probability} and
\emph{closure error probability} are
\begin{align}
  P_e^{(n)}
  &:=\max_{m\in\mathcal M}
  \Pr[\hat{\mathsf S}_o^{(m)}\neq m],
  \label{eq:hamming-error-def}\\
  P_{e,\Cn}^{(n)}
  &:=\max_{m\in\mathcal M}
  \Pr[d_{\Cn}(m,\hat{\mathsf S}_o^{(m)}\mid S_O)>0].
  \label{eq:closure-error-def}
\end{align}
\end{definition}

\begin{remark}[Block source interpretation]
\label{rem:block-source-interpretation}
In the separation results below, a length-\(m\) source block means
\(m\) i.i.d.\ draws from the same deductive source \((S_O,P_O)\),
with the knowledge base \(S_O\) fixed throughout the block.
Thus the channel-coding consequences are standard source--channel
consequences for a structured but stationary finite-alphabet source.
\end{remark}

\begin{theorem}[Deductive source--channel separation]
\label{thm:sem-source-channel}
Let \(W\) be a discrete memoryless channel with capacity \(C(W)\).
Let an i.i.d.\ block of \(m\) source symbols from
\((S_O,P_O)\) be transmitted over \(n\) channel uses, and write
\(\kappa:=n/m\).

For every bounded single-letter distortion family considered in
this paper, the classical separation principle applies:
\begin{enumerate}[label=\textup{(\roman*)}]
  \item if
        \[
          R_{\mathrm{sem}}(D;\,d_{\Cn},\,P_O)
          < \kappa\,C(W),
        \]
        then distortion level \(D\) under closure distortion is
        achievable;
  \item if a distortion level \(D\) under closure distortion is
        achievable, then
        \[
          R_{\mathrm{sem}}(D;\,d_{\Cn},\,P_O)
          \le \kappa\,C(W);
        \]
  \item for every fixed \(\delta\ge 0\), if
        \[
          R_{\mathrm{sem}}(D,\delta;\,d_{\Cn}^{\delta},\,P_O)
          < \kappa\,C(W),
        \]
        then distortion level \(D\) under the bounded-inference
        distortion \(d_{\Cn}^{\delta}\) is achievable; conversely,
        if that distortion level is achievable, then
        \[
          R_{\mathrm{sem}}(D,\delta;\,d_{\Cn}^{\delta},\,P_O)
          \le \kappa\,C(W).
        \]
\end{enumerate}

In particular, in the nondegenerate case \(A\neq\varnothing\), whenever
the hypotheses of Theorem~\ref{thm:tight-zero-rate} for
\(d_{\Cn}\) and of Theorem~\ref{thm:rate-depth} for the relevant
\(\delta\) are satisfied, one has:
\begin{align}
  P_AH(\pi_A) < \kappa\,C(W)
  &\;\Longrightarrow\;
  \text{zero closure distortion is achievable,}
  \label{eq:sep-zero-unconstrained-suf}\\
  \text{if zero closure distortion is achievable, then}\qquad
  P_AH(\pi_A) &\le \kappa\,C(W),
  \label{eq:sep-zero-unconstrained-nec}\\
  P_\delta H(\pi_\delta) < \kappa\,C(W)
  &\;\Longrightarrow\;
  \text{zero \(\delta\)-distortion is achievable,}
  \label{eq:sep-zero-depth-suf}\\
  \text{if zero \(\delta\)-distortion is achievable, then}\qquad
  P_\delta H(\pi_\delta) &\le \kappa\,C(W).
  \label{eq:sep-zero-depth-nec}
\end{align}

More generally, under the hypotheses of Theorem~\ref{thm:confusable-hypergraph-exact},
zero closure distortion is achievable whenever
\[
  P_A\,H_{\Gamma_0}(\pi_A) < \kappa\,C(W),
\]
and is possible only if \(P_A\,H_{\Gamma_0}(\pi_A)\le \kappa\,C(W)\).
\end{theorem}

\begin{proof}[Proof sketch]
Each distortion family considered here is a bounded single-letter
distortion on a finite alphabet. Therefore the standard
source--channel separation theorem for discrete memoryless sources and
channels applies.
For \(d_{\Cn}\), the relevant source-coding function is
\(R_{\mathrm{sem}}(D;\,d_{\Cn},\,P_O)\); for
\(d_{\Cn}^{\delta}\), it is
\(R_{\mathrm{sem}}(D,\delta;\,d_{\Cn}^{\delta},\,P_O)\).
The zero-distortion thresholds follow by substituting the exact
expressions from Theorems~\ref{thm:tight-zero-rate},
\ref{thm:confusable-hypergraph-exact}, and~\ref{thm:rate-depth}
into the separation inequalities.
Full details are given in Appendix~\ref{app:proof-separation}.
\end{proof}

\begin{corollary}[Depth-budget thresholds from separation]
\label{cor:semantic-nyquist}
Assume \(S_O\subseteq\hat S_O\),
Assumption~\ref{assump:atom-equals-essential}, and
Assumption~\ref{assump:delta-core-disjoint} for every \(\delta\ge 0\).
Fix \(\kappa=n/m\) and a channel \(W\) with capacity \(C(W)\). Define
\begin{align}
  \delta_{\mathrm{ach}}^*(S_O,P_O,W,\kappa)
  &:=
  \min\{\delta\ge 0:\ P_\delta H(\pi_\delta)<\kappa C(W)\},
  \label{eq:min-depth-ach}\\
  \delta_{\mathrm{nec}}^*(S_O,P_O,W,\kappa)
  &:=
  \min\{\delta\ge 0:\ P_\delta H(\pi_\delta)\le\kappa C(W)\},
  \label{eq:min-depth-nec}
\end{align}
with the convention \(\min\varnothing:=\infty\).

Then:
\begin{enumerate}[label=\textup{(\roman*)}]
  \item every \(\delta\ge \delta_{\mathrm{ach}}^*\) is sufficient for
        zero \(\delta\)-distortion transmission;
  \item if zero \(\delta\)-distortion transmission is achievable, then
        necessarily \(\delta\ge \delta_{\mathrm{nec}}^*\);
  \item one always has
        \[
          \delta_{\mathrm{nec}}^*
          \le
          \delta_{\mathrm{ach}}^*;
        \]
  \item if no \(\delta\) satisfies
        \(P_\delta H(\pi_\delta)=\kappa C(W)\), then
        \[
          \delta_{\mathrm{ach}}^*
          =
          \delta_{\mathrm{nec}}^*;
        \]
  \item if \(H(P_O)<\kappa C(W)\), then
        \[
          \delta_{\mathrm{ach}}^*=\delta_{\mathrm{nec}}^*=0;
        \]
  \item if \(P_AH(\pi_A)>\kappa C(W)\), then no finite depth budget
        suffices;
  \item if \(P_AH(\pi_A)<\kappa C(W)< H(P_O)\), then
        \[
          1\le \delta_{\mathrm{nec}}^*
          \le \delta_{\mathrm{ach}}^*
          \le \mathsf{D_d}.
        \]
\end{enumerate}
\end{corollary}

\begin{proof}
By Corollary~\ref{cor:rate-depth-endpoints}, the function
\[
  \phi(\delta):=P_\delta H(\pi_\delta)
\]
is non-increasing in \(\delta\) and stabilizes at
\(P_AH(\pi_A)\) for all \(\delta\ge \mathsf{D_d}\).
Part~\textup{(i)} follows from
Theorem~\ref{thm:sem-source-channel}\textup{(iii)} together with the
definition of \(\delta_{\mathrm{ach}}^*\).
Part~\textup{(ii)} follows from the converse direction in
Theorem~\ref{thm:sem-source-channel}\textup{(iii)} and the definition
of \(\delta_{\mathrm{nec}}^*\).
Part~\textup{(iii)} is immediate because
\[
  \{\delta:\phi(\delta)<\kappa C(W)\}
  \subseteq
  \{\delta:\phi(\delta)\le \kappa C(W)\}.
\]
Part~\textup{(iv)} follows because when equality never occurs, the two
sets above coincide.
Part~\textup{(v)} follows from
\(\phi(0)=H(P_O)\).
Part~\textup{(vi)} follows because
\(\phi(\delta)=P_AH(\pi_A)\) for all \(\delta\ge\mathsf{D_d}\), so if
\(P_AH(\pi_A)>\kappa C(W)\), even the stabilized large-\(\delta\)
endpoint lies above channel capability.
Part~\textup{(vii)} follows because
\(\phi(0)=H(P_O)>\kappa C(W)\), so neither threshold is attained at
\(\delta=0\), whereas
\(\phi(\mathsf{D_d})=P_AH(\pi_A)<\kappa C(W)\); monotonicity of
\(\phi\) then gives
\[
  1\le \delta_{\mathrm{nec}}^*
  \le \delta_{\mathrm{ach}}^*
  \le \mathsf{D_d}.
\]
\end{proof}

\begin{theorem}[Message-set converse bounds]
\label{thm:converse}
Let \(A=\Atom(S_O)\).
For any \((n,M)\) block code over a discrete memoryless channel
\(W\) with capacity \(C(W)\):
\begin{enumerate}[label=\textup{(\roman*)}]
  \item if \(P_e^{(n)}\le\epsilon\), then
        \begin{equation}\label{eq:classical-converse}
          \log M
          \le
          \frac{nC(W)+1}{1-\epsilon};
        \end{equation}
  \item if \(A\neq\varnothing\), \(\mathcal M=S_O\),
        \(P_{e,\Cn}^{(n)}\le\epsilon\), and
        Assumption~\textup{\ref{assump:core-disjoint}} holds, then
        \begin{equation}\label{eq:converse-closure}
          \log|A|
          \le
          \frac{nC(W)+1}{1-\epsilon}.
        \end{equation}
\end{enumerate}
\end{theorem}

\begin{proof}
Part~\textup{(i)} is the standard Fano converse
\cite[Ch.~2]{cover2006elements}\cite[Ch.~3]{csiszar2011information}.

For~\textup{(ii)}, let \(M_A\) be uniformly distributed on \(A\) and
consider the induced \(|A|\)-message subproblem. Since
\(P_{e,\Cn}^{(n)}\le\epsilon\) in the max-error sense, the average
closure error probability for \(M_A\) is also at most \(\epsilon\)
by definitions of max-error and average error
\cite[Ch.~7]{cover2006elements}\cite[Ch.~6]{csiszar2011information}.
Under Assumption~\ref{assump:core-disjoint}, the zero-distortion
reconstruction sets \(\{R_{\Cn}(a):a\in A\}\) are pairwise disjoint,
so a decoder output that is closure-correct for one core symbol cannot
be closure-correct for another. Thus closure-correct decoding of core
symbols induces an ordinary \(|A|\)-ary message recovery problem
\cite[Ch.~7]{cover2006elements}\cite[Ch.~11]{csiszar2011information},
and Fano's inequality yields~\eqref{eq:converse-closure}.
\end{proof}

\begin{theorem}[Closure-reliable achievability]
\label{thm:achievability}
Let \(W\) be a discrete memoryless channel with capacity \(C(W)\).
Then:
\begin{enumerate}[label=\textup{(\roman*)}]
  \item if \(S_O\subseteq\hat S_O\), then the full source \(S_O\) can
        be transmitted Hamming-reliably whenever
        \[
          \frac{\log|S_O|}{n}<C(W);
        \]
  \item if \(A\neq\varnothing\) and \(A\subseteq\hat S_O\), then the
        full source \(S_O\) can be transmitted
        closure-reliably whenever
        \[
          \frac{\log|A|}{n}<C(W).
        \]
\end{enumerate}
\end{theorem}

\begin{proof}
Part~\textup{(i)} is Shannon channel coding for a message set of
size \(|S_O|\); see, e.g.,
\cite[Ch.~7]{cover2006elements}\cite[Ch.~6]{csiszar2011information}.

For~\textup{(ii)}, code only the core \(A\).
Choose an \((n,|A|)\) reliable channel code for \(A\)
\cite[Ch.~7]{cover2006elements}\cite[Ch.~6]{csiszar2011information}.
If the transmitted symbol is \(a\in A\), decode it as \(a\).
If the transmitted symbol is \(j\in J\), encode it using the
codeword of any fixed \(a_0\in A\), and decode to whatever core
symbol the channel decoder outputs.
Closure correctness for \(a\in A\) follows from correct decoding,
and closure correctness for \(j\in J\) follows from
Proposition~\ref{prop:redundant-free}.
\end{proof}

\begin{remark}[Heuristic operational deductive compression ratio]
\label{rem:min-blocklength}
Assume \(A\neq\varnothing\),
Assumption~\textup{\ref{assump:core-disjoint}}, and
\(S_O\subseteq\hat S_O\).
Ignoring finite-blocklength backoff terms and the additive constant in
Fano's inequality, Theorems~\ref{thm:converse}
and~\ref{thm:achievability} suggest the benchmark blocklengths
\[
  n_H^\dagger
  :=
  \left\lceil \frac{\log|S_O|}{C(W)} \right\rceil,
  \qquad
  n_{\Cn}^\dagger
  :=
  \left\lceil \frac{\log|A|}{C(W)} \right\rceil.
\]
Accordingly, the corresponding heuristic deductive compression ratio is
\[
  \frac{n_{\Cn}^\dagger}{n_H^\dagger}
  \approx
  \frac{\log|A|}{\log|S_O|},
\]
which is strictly smaller than \(1\) whenever \(J\neq\varnothing\).
This interpretation is asymptotic and operational, rather than an
exact finite-blocklength theorem.
\end{remark}

\subsection{Closure-Adapted Converse Bounds}
\label{subsec:converse-bounds}

\begin{theorem}[Closure-adapted Fano-type lower bound]
\label{thm:semantic-fano-tight}
Let \((S_O,P_O)\) be full-support, and assume that the reconstruction
alphabet satisfies
\[
  \hat S_O\subseteq \Cn(S_O)\cap\mathbb{S}_O.
\]
Let \(\hat{\mathsf S}_o\) be any reconstruction random variable on
\(\hat S_O\), and define
\[
  \epsilon_{\Cn}
  :=
  \Pr[d_{\Cn}({\mathsf{S}}_o,\hat {\mathsf{S}}_o\mid S_O)>0].
\]
If \(A=\varnothing\), then \(\epsilon_{\Cn}=0\) and the conclusion is
trivial.
Otherwise,
\[
  \epsilon_{\Cn}
  =
  \sum_{a\in A}P_O(a)\,
  \Pr[\hat {\mathsf{S}}_o\notin R_{\Cn}(a)\mid {\mathsf{S}}_o=a].
\]
Under Assumption~\ref{assump:core-disjoint},
\begin{equation}\label{eq:sem-fano-tight}
  I({\mathsf{S}}_o;\hat {\mathsf{S}}_o)
  \;\ge\;
  P_A\,H(\pi_A)
  -h_b(\epsilon_{\Cn})
  -\epsilon_{\Cn}\log(|A|-1),
\end{equation}
for \(|A|\ge 2\), while for \(|A|=1\) the bound reduces to
\(I({\mathsf{S}}_o;\hat {\mathsf{S}}_o)\ge 0\).

When \(J=\varnothing\), the inequality takes the same Fano form on the
full source alphabet \(S_O\); if, in addition,
\(R_{\Cn}(s)=\{s\}\) for every \(s\in S_O\), it is exactly the
classical Fano inequality.
\end{theorem}

\begin{proof}[Proof sketch]
Because \(\hat S_O\subseteq\Cn(S_O)\), every redundant input is
automatically closure-correct by Proposition~\ref{prop:redundant-free},
so closure errors can arise only on the core \(A\).
Conditioning on the event \(B:=\mathbf{1}[{\mathsf{S}}_o\in A]\), one reduces the
problem to a \(|A|\)-ary source with law \(\pi_A\).
Under Assumption~\ref{assump:core-disjoint}, the sets
\(\{R_{\Cn}(a):a\in A\}\) are pairwise disjoint, so the decoder output
induces an ordinary decision rule on \(A\cup\{\mathsf e\}\).
Applying the classical Fano inequality to the conditioned core source
and then multiplying by \(P_A\) yields the bound.
Full details are given in
Appendix~\ref{app:proof-semantic-fano-tight}.
\end{proof}

\begin{remark}[Why the bound is sharper than the classical one]
\label{rem:fano-operational}
The reference term in
\eqref{eq:sem-fano-tight} is \(P_AH(\pi_A)\), not \(H(P_O)\).
Thus the converse is aligned with the effective source complexity
under closure fidelity rather than with the raw symbol entropy of
the stored source.
This is precisely the right notion of converse once redundant
states are free.
\end{remark}

\subsection{One Heterogeneous Extension}
\label{subsec:heterogeneous-extension}

We now consider a receiver whose reconstruction alphabet is a
restricted subset \(V\subseteq\mathbb{S}_O\).
This captures vocabulary mismatch or heterogeneous receivers.
The question is whether the exact zero-distortion rate remains the
same when the decoder is no longer allowed to reproduce arbitrary
elements of the sender's alphabet.

\begin{definition}[Restricted-alphabet semantic rate--distortion function]
\label{def:restricted-rd}
For \(V\subseteq\mathbb{S}_O\), define
\begin{equation}\label{eq:restricted-rd}
  R_{\mathrm{sem}}^{(V)}(D;\,d_{\Cn},\,P_O)
  :=
  \min_{\substack{P_{\hat S\mid S}:\\
  \hat {\mathsf{S}}_o\in V\ \text{a.s.}\\
  \E[d_{\Cn}({\mathsf{S}}_o,\hat {\mathsf{S}}_o)]\le D}}
  I(\mathsf{S}_o;\hat{\mathsf{S}}_o).
\end{equation}
\end{definition}

\begin{remark}[Fully redundant restricted-alphabet case]
\label{rem:fully-redundant-restricted}
If \(A=\varnothing\) and \(V\cap\Cn(S_O)\neq\varnothing\), then a
constant reconstruction to any \(\hat s^\star\in V\cap\Cn(S_O)\)
achieves zero closure distortion, so
\[
  R_{\mathrm{sem}}^{(V)}(0;\,d_{\Cn},\,P_O)=0.
\]
The theorem below therefore focuses on the nondegenerate case
\(A\neq\varnothing\).
\end{remark}

\begin{theorem}[Heterogeneous zero-distortion extension]
\label{thm:heterogeneous-extension}
Let \((S_O,P_O)\) be a deductive source with nonempty core
\(A=\Atom(S_O)\), and let \(V\subseteq\mathbb{S}_O\) be a
restricted reconstruction alphabet.
For each \(a\in A\), define
\[
  R_{\Cn}^{(V)}(a)
  :=
  \{\hat s\in V:\ d_{\Cn}(a,\hat s\mid S_O)=0\}.
\]
Assume that:
\begin{enumerate}[label=\textup{(H\arabic*)}]
  \item \(R_{\Cn}^{(V)}(a)\neq\varnothing\) for every \(a\in A\);
  \item for distinct \(a_1,a_2\in A\),
        \[
          R_{\Cn}^{(V)}(a_1)\cap R_{\Cn}^{(V)}(a_2)=\varnothing.
        \]
\end{enumerate}
Then
\begin{equation}\label{eq:heterogeneous-zero-rate}
  R_{\mathrm{sem}}^{(V)}(0;\,d_{\Cn},\,P_O)
  =
  P_A\,H(\pi_A),
\end{equation}
with the convention that the right-hand side is \(0\) when
\(P_A=0\).

Conversely, if there exists \(a\in A\) with \(P_O(a)>0\) and
\(R_{\Cn}^{(V)}(a)=\varnothing\), then the zero-distortion feasible
set is empty, i.e.,
\[
  R_{\mathrm{sem}}^{(V)}(0;\,d_{\Cn},\,P_O)=\infty.
\]
\end{theorem}

\begin{proof}[Proof sketch]
If \(P_A=0\), the claim is immediate by constant reconstruction to any
representative of a core symbol guaranteed by \textup{(H1)}. Assume
henceforth that \(P_A>0\).

Choose, for each \(a\in A\), one representative
\(\phi(a)\in R_{\Cn}^{(V)}(a)\).
Assumption \textup{(H2)} makes \(\phi\) injective, so the core can be
coded exactly as in Proposition~\ref{prop:zero-rate-ach}, while the
redundant part can be mapped to a common distribution supported on
\(\phi(A)\). This gives zero distortion and mutual information
\(P_AH(\pi_A)\).

For the converse, any zero-distortion test channel supported on \(V\)
must map each \(a\in A\) into \(R_{\Cn}^{(V)}(a)\), and these sets are
pairwise disjoint by \textup{(H2)}. Thus the converse proof of
Theorem~\ref{thm:tight-zero-rate} applies verbatim after replacing
\(R_{\Cn}(a)\) by \(R_{\Cn}^{(V)}(a)\). Full details are given in
Appendix~\ref{app:proof-heterogeneous-extension}.
\end{proof}

\begin{corollary}[Core coverage as a simple sufficient condition]
\label{cor:heterogeneous-core-coverage}
If \(A\subseteq V\) and Assumption~\ref{assump:core-disjoint}
holds, then
\[
  R_{\mathrm{sem}}^{(V)}(0;\,d_{\Cn},\,P_O)
  =
  P_AH(\pi_A),
\]
with the convention that the right-hand side is \(0\) when
\(P_A=0\).
\end{corollary}

\begin{proof}
If \(A\subseteq V\), then each \(a\in A\) itself belongs to
\(R_{\Cn}^{(V)}(a)\), so \textup{(H1)} holds.
Assumption~\ref{assump:core-disjoint} implies \textup{(H2)} after
restriction to \(V\).
Apply Theorem~\ref{thm:heterogeneous-extension}.
\end{proof}

\begin{remark}[Meaning of the heterogeneous extension]
\label{rem:heterogeneous-meaning}
Theorem~\ref{thm:heterogeneous-extension} shows that vocabulary
heterogeneity matters only through the availability of
zero-distortion representatives for core symbols.
If every core element still has a distinct closure-equivalent
representative in the receiver's alphabet, then the exact
zero-distortion rate is unchanged.
If even one positive-probability core element has no such
representative, zero distortion becomes impossible.
\end{remark}

\section{Conclusion}
\label{sec:conclusion}

This paper studied lossy compression of a finite knowledge source equipped with a fixed proof system, with fidelity measured by preservation of deductive closure rather than by symbolwise equality. The central conclusion is that, under closure fidelity, the effective source complexity is governed by the irredundant core of the knowledge base rather than by the full stored source. Under an explicit disjointness condition on zero-distortion reconstruction sets, we showed that the zero-distortion rate is exactly \(P_A H(\pi_A)\), and that the full rate--distortion function decomposes into a contribution from
that core alone whenever the reconstruction alphabet is contained in
\(\Cn(S_O)\). We further derived an exact rate--depth--distortion characterization
under a bounded inference-depth budget; under an
additional order-robustness assumption ensuring that the chosen core
coincides with the essential set, this yields a sharp interpolation
between classical symbolwise compression and unconstrained deductive
compression.

From an information-theoretic perspective, the paper contributes a
source-model formulation for deductive compression rather than merely
a semantic communication architecture. Its main insight is that shared inference structure can be treated as a rate-reducing resource: redundant stored states are free under closure fidelity because they can be re-derived at the decoder. Source--channel separation statements, closure-adapted converse bounds, and heterogeneous-receiver extensions follow from this source-model viewpoint.

Several directions remain open.
On the combinatorial side, the hypergraph/graph-entropy characterisation of the
confusable-core zero-distortion limit raises algorithmic questions: computing or
approximating \(H_{\Gamma_0}(\pi_A)\) and its graph-entropy specialisations, and
identifying logical/proof-theoretic conditions under which the hypergraph
collapses to a graph (or further to a Shannon-entropy closed form).

On the information-theoretic side, a natural challenge is to extend the same
confusability-based viewpoint beyond \(D=0\) under closure-based (non-binary)
distortions, and to go beyond finite-alphabet and i.i.d.\ models.
Other natural directions are multiterminal and interactive versions of deductive
compression, and settings with heterogeneous proof systems across agents.
We hope that the formulation developed here provides a useful starting point for
a broader information-theoretic theory of structured sources with inference-enabled
reconstruction.


\section*{Acknowledgment}

The author thanks Associate Professor Rui Wang of the School of
Computer Science at Shanghai Jiao Tong University for insightful
comments, and doctoral students Yiming Wang, Chun Li, Hu Xu,
Siyuan Qiu, Zeyan Li, Jiashuo Zhang, Junxuan He, and Xiao Wang
for stimulating discussions throughout the writing and revision of
this paper.

\appendices

\counterwithin{definition}{section}
\counterwithin{axiom}{section}
\counterwithin{assumption}{section}
\counterwithin{theorem}{section}
\counterwithin{lemma}{section}
\counterwithin{proposition}{section}
\counterwithin{corollary}{section}
\counterwithin{remark}{section}
\counterwithin{example}{section}

\section{Proofs of Sections III--VII}
\label{app:proofs}

This appendix collects deferred proofs for the main results stated in
Sections~III--VII of the paper.

\subsection{Proof of Theorem~\ref{thm:tight-zero-rate}}
\label{app:proof-tight-zero-rate}

\begin{proof}
If \(P_A=0\), then the claim follows from
Remark~\ref{rem:zero-mass-core-case}, so assume henceforth that
\(P_A>0\).

The achievability part is Proposition~\ref{prop:zero-rate-ach}, so
it remains to prove the converse.

Let \(P_{\hat S\mid S}\) be any conditional distribution achieving
zero expected closure distortion. Define
\[
  T(s_o)
  :=
  \begin{cases}
    s_o, & s_o\in A,\\
    *,   & s_o\in J.
  \end{cases}
\]
Since \(T\) is a deterministic function of \({\mathsf{S}}_o\), the data
processing inequality \cite[Ch.~2]{cover2006elements} gives
\[
  I({\mathsf{S}}_o;\hat {\mathsf{S}}_o)\ge I(T;\hat {\mathsf{S}}_o).
\]

For each \(a\in A\), the zero-distortion constraint implies that
\(P_{\hat S\mid S}(\cdot\mid a)\) is supported on \(R_{\Cn}(a)\).
By Assumption~\ref{assump:core-disjoint}, these supports are
pairwise disjoint across \(a\in A\).
For \(j\in J\), no analogous distinguishability constraint is
required because redundant states are free under
Proposition~\ref{prop:redundant-free}.

Write
\[
  P_a:=P_{\hat S\mid S}(\cdot\mid a),
  \qquad
  \bar\pi:=\sum_{a\in A}\pi_A(a)\,P_a.
\]
Since the supports of \(\{P_a\}_{a\in A}\) are pairwise disjoint,
the chain rule for entropy \cite[Ch.~2]{cover2006elements}\cite[Ch.~3]{csiszar2011information} yields
\[
  H(\bar\pi)
  =
  H(\pi_A)+\sum_{a\in A}\pi_A(a)\,H(P_a).
\]

If \(P_J>0\), then define
\[
  Q:=\sum_{j\in J}\frac{P_O(j)}{P_J}\,
  P_{\hat S\mid S}(\cdot\mid j).
\]
The output marginal is \(P_A\bar\pi+P_JQ\), and concavity \cite[Ch.~2]{cover2006elements}\cite[Ch.~3]{csiszar2011information} of entropy
gives
\[
  H(P_A\bar\pi+P_JQ)
  \ge
  P_AH(\bar\pi)+P_JH(Q).
\]
Hence
\begin{align*}
  I(T;\hat {\mathsf{S}}_o)
  &= H(P_A\bar\pi+P_JQ)
     -P_A\sum_{a\in A}\pi_A(a)H(P_a)
     -P_JH(Q) \\
  &\ge
     P_A\Bigl(
       H(\bar\pi)-\sum_{a\in A}\pi_A(a)H(P_a)
     \Bigr) \\
  &= P_A\,H(\pi_A).
\end{align*}

If \(P_J=0\), then \(P_A=1\) and the output marginal is simply
\(\bar\pi\). Therefore
\[
  I(T;\hat {\mathsf{S}}_o)
  =
  H(\bar\pi)-\sum_{a\in A}\pi_A(a)H(P_a)
  =
  H(\pi_A)
  =
  P_AH(\pi_A).
\]

Thus in all cases
\[
  I({\mathsf{S}}_o;\hat {\mathsf{S}}_o)\ge I(T;\hat {\mathsf{S}}_o)\ge P_AH(\pi_A).
\]
Combining this with Proposition~\ref{prop:zero-rate-ach} proves
\[
  R_{\mathrm{sem}}(0;\,d_{\Cn},\,P_O)=P_AH(\pi_A).
\]
\end{proof}

\subsection{Proof of Theorem~\ref{thm:confusable-hypergraph-exact}}
\label{app:proof-confusable-hypergraph-exact}

\begin{proof}
If \(P_A=0\), then the source is supported on \(J=S_O\setminus A\).
By Assumption~\ref{assump:core-coverage} there exists
\(\hat s^\star\in\hat S_O\cap\Cn(S_O)\), and the constant reconstruction
\(\hat{\mathsf{S}}_o\equiv \hat s^\star\) achieves zero distortion by
Proposition~\ref{prop:redundant-free}. Hence
\(R_{\mathrm{sem}}(0;\,d_{\Cn},\,P_O)=0\), matching
\eqref{eq:R0-hypergraph}.

Assume \(P_A>0\).
Let \(\mathsf{S}_o\sim P_O\), let \(\hat{\mathsf{S}}_o\) be generated by a
test channel \(P_{\hat S\mid S}\), and assume
\(\mathbb{E}[d_{\Cn}(\mathsf{S}_o,\hat{\mathsf{S}}_o)]=0\).
Since \(d_{\Cn}\ge 0\), this implies \(d_{\Cn}(\mathsf{S}_o,\hat{\mathsf{S}}_o)=0\) a.s.

Let \(\mathsf{B}:=\mathbf{1}[\mathsf{S}_o\in A]\) and let \(\mathsf{A}_o\) denote
\(\mathsf{S}_o\) conditioned on \(\mathsf{B}=1\), so \(\mathsf{A}_o\sim\pi_A\).

\smallskip\noindent
\emph{Converse.}
By the chain rule and nonnegativity of mutual information,
\[
  I(\mathsf{S}_o;\hat{\mathsf{S}}_o)
  =
  I(\mathsf{B};\hat{\mathsf{S}}_o)+I(\mathsf{S}_o;\hat{\mathsf{S}}_o\mid \mathsf{B})
  \ge
  P_A\,I(\mathsf{A}_o;\hat{\mathsf{S}}_o).
\]
Under \(\mathsf{B}=1\), we have \(\hat{\mathsf{S}}_o\in R_{\Cn}(\mathsf{A}_o)\) a.s.
Define a random subset \(W\subseteq A\) as a deterministic function of \(\hat{\mathsf{S}}_o\):
\[
  W:=\{a\in A:\ \hat{\mathsf{S}}_o\in R_{\Cn}(a)\}.
\]
Then \(\mathsf{A}_o\in W\) a.s., and \(\hat{\mathsf{S}}_o\in\cap_{a\in W}R_{\Cn}(a)\),
so \(W\in\Gamma_0\) a.s.
Hence \(W\) is feasible for the minimisation defining \(H_{\Gamma_0}(\pi_A)\), so
\(I(\mathsf{A}_o;W)\ge H_{\Gamma_0}(\pi_A)\).
Since \(W\) is a function of \(\hat{\mathsf{S}}_o\), data processing gives
\(I(\mathsf{A}_o;\hat{\mathsf{S}}_o)\ge I(\mathsf{A}_o;W)\).
Combining,
\[
  I(\mathsf{S}_o;\hat{\mathsf{S}}_o)
  \ge
  P_A\,I(\mathsf{A}_o;\hat{\mathsf{S}}_o)
  \ge
  P_A\,I(\mathsf{A}_o;W)
  \ge
  P_A\,H_{\Gamma_0}(\pi_A).
\]
Taking the minimum over all zero-distortion test channels yields
\(R_{\mathrm{sem}}(0)\ge P_A H_{\Gamma_0}(\pi_A)\).

\smallskip\noindent
\emph{Achievability.}
Fix \(\varepsilon>0\).
Let \(P^{\star}_{W\mid A}\) be \(\varepsilon\)-optimal in \eqref{eq:hypergraph-entropy-Gamma0}, i.e.,
\(\mathsf{A}_o\in W\) and \(W\in\Gamma_0\) a.s., and
\(I(\mathsf{A}_o;W)\le H_{\Gamma_0}(\pi_A)+\varepsilon\).

For each \(w\in\Gamma_0\), choose a selector
\[
  \psi(w)\in \bigcap_{a\in w}R_{\Cn}(a),
\]
which exists by definition of \(\Gamma_0\).
Define a test channel as follows.
If \(s_o=a\in A\), sample \(W\sim P^{\star}_{W\mid A}(\cdot\mid a)\) and output
\(\hat s:=\psi(W)\).
If \(s_o=j\in J\), output \(\hat s\) according to the same output marginal as for core inputs,
namely sample \(\tilde A\sim\pi_A\), then \(W\sim P^{\star}_{W\mid A}(\cdot\mid \tilde A)\), and output \(\psi(W)\).

\emph{Zero distortion.}
If \(s_o=a\in A\), then \(a\in W\) a.s. and \(\psi(W)\in R_{\Cn}(a)\), hence
\(d_{\Cn}(a,\psi(W)\mid S_O)=0\).
If \(s_o=j\in J\), pick any \(a\in W\) (possible since \(W\neq\emptyset\) a.s.).
Then \(\psi(W)\in R_{\Cn}(a)\), so \(d_{\Cn}(a,\psi(W)\mid S_O)=0\), and thus
\(\psi(W)\in\Cn(S_O)\) by Lemma~\ref{lem:zero-distortion-implies-in-closure}.
Therefore \(d_{\Cn}(j,\psi(W)\mid S_O)=0\) by Proposition~\ref{prop:redundant-free}.
Hence the channel achieves zero expected distortion.

\emph{Rate.}
By construction, the output distribution is the same whether \(\mathsf{B}=1\) or \(\mathsf{B}=0\),
so \(I(\mathsf{B};\hat{\mathsf{S}}_o)=0\), and
\(I(\mathsf{S}_o;\hat{\mathsf{S}}_o)=P_A I(\mathsf{A}_o;\hat{\mathsf{S}}_o)\).
Moreover, under \(\mathsf{B}=1\) the channel factors as the Markov chain
\(\mathsf{A}_o\to W\to \hat{\mathsf{S}}_o=\psi(W)\), hence
\(I(\mathsf{A}_o;\hat{\mathsf{S}}_o)\le I(\mathsf{A}_o;W)\) by data processing.
Therefore
\[
  I(\mathsf{S}_o;\hat{\mathsf{S}}_o)
  \le
  P_A\,I(\mathsf{A}_o;W)
  \le
  P_A\,(H_{\Gamma_0}(\pi_A)+\varepsilon).
\]
Letting \(\varepsilon\downarrow 0\) yields \(R_{\mathrm{sem}}(0)\le P_A H_{\Gamma_0}(\pi_A)\),
which together with the converse proves \eqref{eq:R0-hypergraph}.
\end{proof}

\subsection{Proof of Theorem~\ref{thm:rd-decomposition}}
\label{app:proof-rd-decomposition}

\begin{proof}
If \(P_A=0\), then the source is supported on \(J=S_O\setminus A\).
Because \(\hat S_O\) is nonempty and satisfies
\(\hat S_O\subseteq\Cn(S_O)\), choose any fixed
\(\hat s^\star\in\hat S_O\).
For every \(j\in J\), Proposition~\ref{prop:redundant-free} gives
\[
  d_{\Cn}(j,\hat s^\star\mid S_O)=0.
\]
Hence the constant reconstruction \(\hat {\mathsf{S}}_o\equiv\hat s^\star\)
achieves zero distortion almost surely and therefore zero rate.
Thus both sides of~\eqref{eq:rd-decomp} equal \(0\).
Assume \(P_A>0\).

\smallskip\noindent
\emph{Step 1: distortion depends only on the core.}
Since \(\hat S_O\subseteq\Cn(S_O)\), Proposition~\ref{prop:redundant-free}
implies
\[
  d_{\Cn}(j,\hat s\mid S_O)=0,
  \qquad
  \forall\,j\in J,\ \forall\,\hat s\in\hat S_O.
\]
Therefore, for every test channel \(P_{\hat S\mid S}\),
\[
  \E[d_{\Cn}({\mathsf{S}}_o,\hat {\mathsf{S}}_o)]
  =
  P_A\,
  \E[d_{\Cn}({\mathsf{S}}_o,\hat {\mathsf{S}}_o)\mid {\mathsf{S}}_o\in A].
\]
Thus the global distortion constraint
\(\E[d_{\Cn}({\mathsf{S}}_o,\hat {\mathsf{S}}_o)]\le D\)
is equivalent to
\[
  \E[d_{\Cn}({\mathsf{S}}_o,\hat {\mathsf{S}}_o)\mid {\mathsf{S}}_o\in A]
  \le
  \frac{D}{P_A}.
\]

\smallskip\noindent
\emph{Step 2: lower bound on mutual information.}
Define
\[
  T(s_o)
  :=
  \begin{cases}
    s_o, & s_o\in A,\\
    *,   & s_o\in J.
  \end{cases}
\]
Then \(T\) is a deterministic function of \({\mathsf{S}}_o\), the data processing inequality \cite[Ch.~2]{cover2006elements} yields
\[
  I({\mathsf{S}}_o;\hat {\mathsf{S}}_o)\ge I(T;\hat {\mathsf{S}}_o).
\]

Let \(P_a:=P_{\hat S\mid S}(\cdot\mid a)\) for \(a\in A\), and let
\[
  \bar P_A:=\sum_{a\in A}\pi_A(a)\,P_a.
\]
The quantity
\[
  H(\bar P_A)-\sum_{a\in A}\pi_A(a)H(P_a)
\]
is exactly
\(I(\mathsf{A}_o;\hat{\mathsf{S}}_o)\)
for the core sub-source \(\mathsf{A}_o\sim\pi_A\) and test channel
\(P_{\hat S\mid A}\).

If \(0<P_A<1\), let
\[
  Q:=\sum_{j\in J}\frac{P_O(j)}{P_J}\,
  P_{\hat S\mid S}(\cdot\mid j).
\]
Then the output marginal is
\[
  P_{\hat S}=P_A\bar P_A+P_JQ.
\]
By concavity of entropy \cite[Ch.~2]{cover2006elements}\cite[Ch.~3]{csiszar2011information},
\[
  H(P_{\hat S})\ge P_AH(\bar P_A)+P_JH(Q).
\]
Also,
\[
  H(\hat {\mathsf{S}}_o\mid T)
  =
  P_A\sum_{a\in A}\pi_A(a)H(P_a)+P_JH(Q).
\]
Hence
\begin{align*}
  I(T;\hat {\mathsf{S}}_o)
  &=H(P_{\hat S})-H(\hat {\mathsf{S}}_o\mid T) \\
  &\ge
  P_A\Bigl(
    H(\bar P_A)-\sum_{a\in A}\pi_A(a)H(P_a)
  \Bigr) \\
  &= P_A\,I(\mathsf{A}_o;\hat{\mathsf{S}}_o).
\end{align*}

If \(P_A=1\), then \(J=\varnothing\) and \(T={\mathsf{S}}_o\) almost surely.
The output marginal is simply \(\bar P_A\), so
\begin{align*}
  I(T;\hat {\mathsf{S}}_o)
  &=
  H(\bar P_A)-\sum_{a\in A}\pi_A(a)H(P_a) \\
  &=
  I(\mathsf{A}_o;\hat{\mathsf{S}}_o)
  =
  P_A\,I(\mathsf{A}_o;\hat{\mathsf{S}}_o).
\end{align*}

Therefore every feasible global test channel satisfies
\[
  I({\mathsf{S}}_o;\hat {\mathsf{S}}_o)
  \ge
  P_A\,
  R^{(A)}\!\Bigl(\frac{D}{P_A};\,d_{\Cn},\,\pi_A\Bigr).
\]

\smallskip\noindent
\emph{Step 3: achievability by homogenizing the redundant part.}
Let \(P_{\hat S\mid A}^*\) achieve
\(R^{(A)}(D/P_A;\,d_{\Cn},\,\pi_A)\)
\cite[Ch.~10]{cover2006elements}
\cite[Ch.~7]{csiszar2011information}, and let
\[
  \bar P_A^*
  :=
  \sum_{a\in A}\pi_A(a)\,P_{\hat S\mid A}^*(\cdot\mid a)
\]
be the corresponding output marginal.
Extend \(P_{\hat S\mid A}^*\) to a test channel on all of \(S_O\) by
\[
  P_{\hat S\mid S}(\cdot\mid s_o)
  :=
  \begin{cases}
    P_{\hat S\mid A}^*(\cdot\mid s_o), & s_o\in A,\\[3pt]
    \bar P_A^*(\cdot), & s_o\in J.
  \end{cases}
\]
Because every \(j\in J\) is distortion-free against every
\(\hat s\in\hat S_O\subseteq\Cn(S_O)\), this extension preserves
the distortion level \(D\).

Moreover, the overall output marginal is exactly \(\bar P_A^*\), so
\[
  H(\hat {\mathsf{S}}_o)=H(\bar P_A^*),
\]
and
\[
  H(\hat {\mathsf{S}}_o\mid {\mathsf{S}}_o)
  =
  P_AH(\hat {\mathsf{S}}_o\mid \mathsf{A}_o)+P_JH(\bar P_A^*).
\]
Therefore
\begin{align*}
  I({\mathsf{S}}_o;\hat {\mathsf{S}}_o)
  &=
  H(\bar P_A^*)-
  P_AH(\hat {\mathsf{S}}_o\mid \mathsf{A}_o)-P_JH(\bar P_A^*) \\
  &=
  P_A\Bigl(H(\bar P_A^*)-H(\hat {\mathsf{S}}_o\mid \mathsf{A}_o)\Bigr) \\
  &=
  P_A\,I(\mathsf{A}_o;\hat{\mathsf{S}}_o),
\end{align*}
which equals
\[
  P_A\,
  R^{(A)}\!\Bigl(\frac{D}{P_A};\,d_{\Cn},\,\pi_A\Bigr).
\]
This proves~\eqref{eq:rd-decomp}.
\end{proof}

\subsection{Proof of Theorem~\ref{thm:rate-depth}}
\label{app:proof-rate-depth}

\begin{proof}
Fix \(\delta\ge 0\).

If \(P_\delta=0\), then the source is supported on
\(J_\delta:=S_O\setminus A_\delta\), and
Lemma~\ref{lem:delta-free-substitution} implies that every
\(j\in J_\delta\) is zero-distortion free against every reconstruction
symbol. Hence any constant reconstruction in \(\hat S_O\) achieves
zero \(\delta\)-distortion almost surely, so
\[
  R_{\mathrm{sem}}(0,\delta;\,d_{\Cn}^{\delta},\,P_O)=0,
\]
which agrees with the convention \(P_\delta H(\pi_\delta)=0\).

Assume henceforth that \(P_\delta>0\), and write
\(J_\delta:=S_O\setminus A_\delta\).
Let \(P_{\hat S\mid S}\) be any test channel achieving zero expected
\(\delta\)-distortion. Define
\[
  T_\delta(s_o)
  :=
  \begin{cases}
    s_o, & s_o\in A_\delta,\\
    *,   & s_o\in J_\delta.
  \end{cases}
\]
Since \(T_\delta\) is a deterministic function of \({\mathsf{S}}_o\),
by the data processing inequality \cite[Ch.~2]{cover2006elements}, we have
\[
  I({\mathsf{S}}_o;\hat {\mathsf{S}}_o)\ge I(T_\delta;\hat {\mathsf{S}}_o).
\]

For each \(a\in A_\delta\), zero \(\delta\)-distortion implies that
\(P_{\hat S\mid S}(\cdot\mid a)\) is supported on \(R_\delta(a)\).
By Assumption~\ref{assump:delta-core-disjoint}, these supports are
pairwise disjoint across \(a\in A_\delta\).
For \(j\in J_\delta\), the output is unconstrained by
Lemma~\ref{lem:delta-free-substitution}.

Write
\[
  P_a:=P_{\hat S\mid S}(\cdot\mid a),
  \qquad
  \bar\pi_\delta:=\sum_{a\in A_\delta}\pi_\delta(a)\,P_a.
\]
Since the supports of \(\{P_a\}_{a\in A_\delta}\) are pairwise
disjoint, the chain rule for entropy \cite[Ch.~2]{cover2006elements}\cite[Ch.~3]{csiszar2011information} gives
\[
  H(\bar\pi_\delta)
  =
  H(\pi_\delta)+
  \sum_{a\in A_\delta}\pi_\delta(a)H(P_a).
\]

If \(0<P_\delta<1\), define
\[
  Q_\delta
  :=
  \sum_{j\in J_\delta}\frac{P_O(j)}{1-P_\delta}\,
  P_{\hat S\mid S}(\cdot\mid j).
\]
Then the output marginal is
\[
  P_{\hat S}
  =
  P_\delta \bar\pi_\delta + (1-P_\delta)Q_\delta.
\]
By concavity of entropy \cite[Ch.~2]{cover2006elements}\cite[Ch.~3]{csiszar2011information},
\[
  H(P_{\hat S})
  \ge
  P_\delta H(\bar\pi_\delta)+(1-P_\delta)H(Q_\delta).
\]
Also,
\[
  H(\hat {\mathsf{S}}_o\mid T_\delta)
  =
  P_\delta\sum_{a\in A_\delta}\pi_\delta(a)H(P_a)
  +(1-P_\delta)H(Q_\delta).
\]
Hence
\begin{align*}
  I(T_\delta;\hat {\mathsf{S}}_o)
  &=
  H(P_{\hat S})-H(\hat {\mathsf{S}}_o\mid T_\delta) \\
  &\ge
  P_\delta\Bigl(
    H(\bar\pi_\delta)-
    \sum_{a\in A_\delta}\pi_\delta(a)H(P_a)
  \Bigr) \\
  &=
  P_\delta H(\pi_\delta).
\end{align*}

If \(P_\delta=1\), then \(J_\delta=\varnothing\) and
\(T_\delta={\mathsf{S}}_o\) almost surely. The output marginal is simply
\(\bar\pi_\delta\), so
\begin{align*}
  I(T_\delta;\hat {\mathsf{S}}_o)
  &=
  H(\bar\pi_\delta)-
  \sum_{a\in A_\delta}\pi_\delta(a)H(P_a) \\
  &=
  H(\pi_\delta)
  =
  P_\delta H(\pi_\delta).
\end{align*}

Therefore, in all cases with \(P_\delta>0\),
\[
  I({\mathsf{S}}_o;\hat {\mathsf{S}}_o)\ge I(T_\delta;\hat {\mathsf{S}}_o)\ge P_\delta H(\pi_\delta).
\]

For achievability, define
\[
  P_{\hat S\mid S}(\hat s\mid s_o)
  :=
  \begin{cases}
    \mathbf{1}[\hat s=s_o], & s_o\in A_\delta,\\[3pt]
    \pi_\delta(\hat s), & s_o\in J_\delta,
  \end{cases}
\]
with \(\pi_\delta\) supported on \(A_\delta\subseteq\hat S_O\).
If \(s_o\in A_\delta\), the reconstruction is exact and hence has
zero \(\delta\)-distortion.
If \(s_o\in J_\delta\), zero \(\delta\)-distortion follows from
Lemma~\ref{lem:delta-free-substitution}.
The output marginal is \(\pi_\delta\), and exactly as in the proof of
Proposition~\ref{prop:zero-rate-ach},
\[
  I({\mathsf{S}}_o;\hat {\mathsf{S}}_o)=P_\delta H(\pi_\delta).
\]
Thus
\[
  R_{\mathrm{sem}}(0,\delta;\,d_{\Cn}^{\delta},\,P_O)
  =
  P_\delta H(\pi_\delta).
\]
\end{proof}

\subsection{Proof of Theorem~\ref{thm:rate-depth-distortion}}
\label{app:proof-rate-depth-distortion}

\begin{proof}
Fix \(\delta\ge 0\).
If \(P_\delta=0\), then the source is supported on
\(S_O\setminus A_\delta\). Since \(\hat S_O\) is nonempty, any
constant reconstruction in \(\hat S_O\) has zero \(\delta\)-distortion
almost surely by Lemma~\ref{lem:delta-free-substitution}. Hence
\[
  R_{\mathrm{sem}}(D,\delta;\,d_{\Cn}^{\delta},\,P_O)=0,
  \qquad \forall\,D\ge 0,
\]
which agrees with the convention in~\eqref{eq:rdd-surface}.
Assume henceforth that \(P_\delta>0\).

By Lemma~\ref{lem:delta-free-substitution}, every
\(j\in S_O\setminus A_\delta\) is distortion-free against every
reconstruction symbol under \(d_{\Cn}^{\delta}\).
Therefore, for every test channel,
\[
  \E[d_{\Cn}^{\delta}({\mathsf{S}}_o,\hat {\mathsf{S}}_o)]
  =
  P_\delta\,
  \E[d_{\Cn}^{\delta}({\mathsf{S}}_o,\hat {\mathsf{S}}_o)\mid {\mathsf{S}}_o\in A_\delta].
\]
Thus the distortion constraint
\(\E[d_{\Cn}^{\delta}]\le D\) reduces to a
core-only constraint with threshold \(D/P_\delta\).

Now define
\[
  T_\delta(s_o)
  :=
  \begin{cases}
    s_o, & s_o\in A_\delta,\\
    *,   & s_o\notin A_\delta.
  \end{cases}
\]
Since \(T_\delta\) is a deterministic function of \({\mathsf{S}}_o\), 
by the data processing inequality \cite[Ch.~2]{cover2006elements},
we have
\[
  I({\mathsf{S}}_o;\hat {\mathsf{S}}_o)\ge I(T_\delta;\hat {\mathsf{S}}_o).
\]
Following the same concavity argument as in the converse of
Theorem~\ref{thm:rd-decomposition}
\cite[Ch.~2]{cover2006elements}\cite[Ch.~1]{csiszar2011information}, we obtain
\[
  I(T_\delta;\hat {\mathsf{S}}_o)
  \ge
  P_\delta\,
  I(\mathsf{A}_{\delta,o};\hat{\mathsf{S}}_o),
\]
where \(\mathsf{A}_{\delta,o}\sim\pi_\delta\) and the channel is the
restriction of \(P_{\hat S\mid S}\) to \(A_\delta\).
Taking the minimum over all feasible channels yields
\[
  R_{\mathrm{sem}}(D,\delta)
  \ge
  P_\delta\,
  R^{(A_\delta)}\!\Bigl(\frac{D}{P_\delta};\,
  d_{\Cn}^{\delta},\,\pi_\delta\Bigr).
\]

For achievability, let
\(P_{\hat S\mid A_\delta}^*\) be an optimal test channel that achieves
the rate-distortion function \(R^{(A_\delta)}(\cdot)\) on the right-hand side
\cite[Ch.~10]{cover2006elements}
\cite[Ch.~7]{csiszar2011information},
and extend it to all of \(S_O\) by assigning to every
\(j\notin A_\delta\) the output marginal induced by
\(P_{\hat S\mid A_\delta}^*\).
This keeps the distortion unchanged because the symbols outside
\(A_\delta\) are \(\delta\)-redundant, and the same entropy
calculation as in Appendix~\ref{app:proof-rd-decomposition} shows that
the resulting mutual information is exactly
\[
  P_\delta\,
  R^{(A_\delta)}\!\Bigl(\frac{D}{P_\delta};\,
  d_{\Cn}^{\delta},\,\pi_\delta\Bigr).
\]
Hence~\eqref{eq:rdd-surface} holds.
\end{proof}

\subsection{Proof of Theorem~\ref{thm:sem-source-channel}}
\label{app:proof-separation}

\begin{proof}
All distortion functions used in the theorem are bounded
single-letter distortions on finite alphabets.
Therefore the standard source--channel separation theorem applies \cite[Chs.~7,10]{cover2006elements}
\cite[Ch.~7]{csiszar2011information}.

For closure distortion \(d_{\Cn}\), the source-coding side is governed
by \(R_{\mathrm{sem}}(D;\,d_{\Cn},\,P_O)\)\cite[Ch.~10]{cover2006elements}
\cite[Ch.~7]{csiszar2011information}, and the channel-coding
side is governed by the Shannon capacity \(C(W)\)\cite[Ch.~7]{cover2006elements}
\cite[Ch.~6]{csiszar2011information}.
Hence a distortion level \(D\) is achievable whenever
\[
  R_{\mathrm{sem}}(D;\,d_{\Cn},\,P_O)<\kappa C(W),
\]
and conversely every achievable distortion level must satisfy
\[
  R_{\mathrm{sem}}(D;\,d_{\Cn},\,P_O)\le \kappa C(W).
\]

The same reasoning applies to the bounded-inference distortion
\(d_{\Cn}^{\delta}\): for every fixed \(\delta\),
\[
  R_{\mathrm{sem}}(D,\delta;\,d_{\Cn}^{\delta},\,P_O)
  < \kappa C(W)
\]
is sufficient for achievability, whereas
\[
  R_{\mathrm{sem}}(D,\delta;\,d_{\Cn}^{\delta},\,P_O)
  \le \kappa C(W)
\]
is necessary.

Finally, under the hypotheses of
Theorems~\ref{thm:tight-zero-rate}, \ref{thm:confusable-hypergraph-exact},
and~\ref{thm:rate-depth}, substituting the corresponding exact
zero-distortion formulas yields the displayed sufficient and necessary
conditions \eqref{eq:sep-zero-unconstrained-suf}--\eqref{eq:sep-zero-depth-nec},
as well as the hypergraph-entropy threshold stated at the end of
Theorem~\ref{thm:sem-source-channel}.
\end{proof}

\subsection{Proof of Theorem~\ref{thm:semantic-fano-tight}}
\label{app:proof-semantic-fano-tight}

\begin{proof}
If \(A=\varnothing\), then by Proposition~\ref{prop:redundant-free}
every input is closure-correct against every
\(\hat s\in\hat S_O\subseteq\Cn(S_O)\), so
\(\epsilon_{\Cn}=0\) and the claim is trivial.
Assume henceforth that \(A\neq\varnothing\).

Because \(\hat S_O\subseteq\Cn(S_O)\), Proposition~\ref{prop:redundant-free}
implies that
\[
  d_{\Cn}(j,\hat s\mid S_O)=0,
  \qquad
  \forall\,j\in J,\ \forall\,\hat s\in\hat S_O.
\]
Hence closure errors can arise only on core inputs, and therefore
\[
  \epsilon_{\Cn}
  =
  \sum_{a\in A}P_O(a)\,
  \Pr[\hat {\mathsf{S}}_o\notin R_{\Cn}(a)\mid {\mathsf{S}}_o=a].
\]

Let \(B:=\mathbf{1}[{\mathsf{S}}_o\in A]\).
Since \(B\) is a deterministic function of \({\mathsf{S}}_o\), by the chain rule for mutual information \cite[Ch.~2]{cover2006elements}\cite[Ch.~3]{csiszar2011information}, we have
\[
  I({\mathsf{S}}_o;\hat {\mathsf{S}}_o)=I(B,{\mathsf{S}}_o;\hat {\mathsf{S}}_o)
  =I(B;\hat {\mathsf{S}}_o)+I({\mathsf{S}}_o;\hat {\mathsf{S}}_o\mid B)
  \ge I({\mathsf{S}}_o;\hat {\mathsf{S}}_o\mid B).
\]
Therefore
\[
  I({\mathsf{S}}_o;\hat {\mathsf{S}}_o)\ge
  P_A I({\mathsf{S}}_o;\hat {\mathsf{S}}_o\mid B=1),
\]
since the conditional mutual information given \(B=0\) is
nonnegative.

Conditioned on \(B=1\), the source is distributed as \(\pi_A\) on
the alphabet \(A\).
If \(|A|=1\), then \(H(\pi_A)=0\), so the desired inequality reduces
to the trivial bound \(I({\mathsf{S}}_o;\hat {\mathsf{S}}_o)\ge 0\).
Assume now that \(|A|\ge 2\).

Under Assumption~\ref{assump:core-disjoint}, define
\(g:\hat S_O\to A\cup\{\mathsf e\}\) by
\[
  g(\hat s)
  :=
  \begin{cases}
    a, & \hat s\in R_{\Cn}(a)\text{ for some }a\in A,\\
    \mathsf e, & \text{otherwise}.
  \end{cases}
\]
Because the sets \(R_{\Cn}(a)\) are pairwise disjoint, \(g\) is
well-defined.
By data processing \cite[Ch.~2]{cover2006elements},
\[
  H({\mathsf{S}}_o\mid \hat {\mathsf{S}}_o,\ B=1)
  \le
  H({\mathsf{S}}_o\mid g(\hat {\mathsf{S}}_o),\ B=1).
\]

Let \(p_A:=\epsilon_{\Cn}/P_A\).
Applying the classical Fano inequality
\cite[Ch.~2]{cover2006elements}\cite[Ch.~3]{csiszar2011information}
to the \(|A|\)-ary source conditioned on \(B=1\) gives
\[
  H({\mathsf{S}}_o\mid g(\hat {\mathsf{S}}_o),\ B=1)
  \le
  h_b(p_A)+p_A\log(|A|-1).
\]
Hence
\[
  I({\mathsf{S}}_o;\hat {\mathsf{S}}_o\mid B=1)
  \ge
  H(\pi_A)-h_b(p_A)-p_A\log(|A|-1).
\]
Multiplying by \(P_A\) yields
\[
  I({\mathsf{S}}_o;\hat {\mathsf{S}}_o)
  \ge
  P_AH(\pi_A)
  -P_A h_b\!\Bigl(\frac{\epsilon_{\Cn}}{P_A}\Bigr)
  -\epsilon_{\Cn}\log(|A|-1).
\]
Finally, Jensen's inequality gives
\[
  P_A h_b\!\Bigl(\frac{\epsilon_{\Cn}}{P_A}\Bigr)
  \le
  h_b(\epsilon_{\Cn}),
\]
which proves
\[
  I({\mathsf{S}}_o;\hat {\mathsf{S}}_o)
  \ge
  P_AH(\pi_A)-h_b(\epsilon_{\Cn})
  -\epsilon_{\Cn}\log(|A|-1).
\]
\end{proof}

\subsection{Proof of Theorem~\ref{thm:heterogeneous-extension}}
\label{app:proof-heterogeneous-extension}

\begin{proof}
If \(P_A=0\), then the source is supported on \(J=S_O\setminus A\).
Since \(A\neq\varnothing\) and \textup{(H1)} holds, choose any
\(a_0\in A\) and any representative
\(\phi(a_0)\in R_{\Cn}^{(V)}(a_0)\).
Because
\[
  d_{\Cn}(a_0,\phi(a_0)\mid S_O)=0,
\]
one has
\[
  \Cn((S_O\setminus\{a_0\})\cup\{\phi(a_0)\})=\Cn(S_O),
\]
and hence by reflexivity \(\phi(a_0)\in\Cn(S_O)\).
Therefore Proposition~\ref{prop:redundant-free} gives
\[
  d_{\Cn}(j,\phi(a_0)\mid S_O)=0,
  \qquad \forall\,j\in J.
\]
So the constant reconstruction
\(\hat {\mathsf{S}}_o\equiv\phi(a_0)\) achieves zero distortion almost surely,
and hence
\[
  R_{\mathrm{sem}}^{(V)}(0;\,d_{\Cn},\,P_O)=0,
\]
which agrees with the stated convention. Assume henceforth that
\(P_A>0\).

For each \(a\in A\), assumption \textup{(H1)} ensures that
\(R_{\Cn}^{(V)}(a)\neq\varnothing\), so choose one representative
\(\phi(a)\in R_{\Cn}^{(V)}(a)\).
By \textup{(H2)}, the map \(\phi:A\to V\) is injective.

Define a test channel by
\[
  P_{\hat S\mid S}(\hat s\mid s)
  :=
  \begin{cases}
    \mathbf{1}[\hat s=\phi(s)], & s\in A,\\[3pt]
    \displaystyle\sum_{a\in A}\pi_A(a)\,\mathbf{1}[\hat s=\phi(a)],
    & s\in J.
  \end{cases}
\]
For \(a\in A\), zero distortion holds because
\(\phi(a)\in R_{\Cn}^{(V)}(a)\) by construction.
For \(j\in J\), every possible output is of the form \(\phi(a)\) with
\(a\in A\).
Since
\[
  d_{\Cn}(a,\phi(a)\mid S_O)=0,
\]
the equality of closures implies
\[
  \Cn((S_O\setminus\{a\})\cup\{\phi(a)\})=\Cn(S_O),
\]
and hence by reflexivity
\(\phi(a)\in \Cn(S_O)\).
Therefore Proposition~\ref{prop:redundant-free} gives
\[
  d_{\Cn}(j,\phi(a)\mid S_O)=0.
\]
So the expected distortion is zero.

Because \(\phi\) is injective, the induced output distribution on
\(\phi(A)\) has entropy \(H(\pi_A)\), and exactly the same
calculation as in Proposition~\ref{prop:zero-rate-ach} yields
\[
  I({\mathsf{S}}_o;\hat {\mathsf{S}}_o)=P_AH(\pi_A).
\]
Hence
\[
  R_{\mathrm{sem}}^{(V)}(0;\,d_{\Cn},\,P_O)\le P_AH(\pi_A).
\]

For the converse, let \(P_{\hat S\mid S}\) be any feasible
zero-distortion test channel supported on \(V\).
Then for every \(a\in A\), the conditional law
\(P_{\hat S\mid S}(\cdot\mid a)\) must be supported on
\(R_{\Cn}^{(V)}(a)\).
By assumption \textup{(H2)}, these supports are pairwise disjoint, so
the converse proof of Theorem~\ref{thm:tight-zero-rate} applies
verbatim with \(R_{\Cn}(a)\) replaced by \(R_{\Cn}^{(V)}(a)\).
Thus
\[
  I({\mathsf{S}}_o;\hat {\mathsf{S}}_o)\ge P_AH(\pi_A),
\]
and therefore
\[
  R_{\mathrm{sem}}^{(V)}(0;\,d_{\Cn},\,P_O)=P_AH(\pi_A).
\]

For the final claim, if some \(a\in A\) with \(P_O(a)>0\) satisfies
\(R_{\Cn}^{(V)}(a)=\varnothing\), then there is no reconstruction
symbol in \(V\) that can reproduce \(a\) with zero closure
distortion.
Hence the zero-distortion feasible set is empty, and
\[
  R_{\mathrm{sem}}^{(V)}(0;\,d_{\Cn},\,P_O)=\infty.
\]
\end{proof}

\section{Datalog Classes Satisfying the Assumptions}
\label{app:datalog-classes}

This appendix records concrete Datalog classes under which the
standing assumptions used in the paper are satisfied.
The goal is not to give the most general logical characterization,
but to isolate natural and checkable sufficient conditions.

\subsection{Finite Active-Domain Datalog}
\label{app:datalog-basic}

\begin{definition}[Finite active-domain Datalog instance]
\label{def:finite-active-domain-datalog}
A \emph{finite active-domain Datalog instance} consists of:
\begin{enumerate}[label=\textup{(\roman*)}]
  \item a function-free Datalog program \(P\);
  \item a finite active domain \(\mathcal D\);
  \item the finite set \(\mathbb{S}_O\) of all ground atoms over
        \(\mathcal D\) whose predicates occur in \(P\);
  \item a finite stored knowledge base \(S_O\subseteq\mathbb{S}_O\).
\end{enumerate}
The proof system is the standard Datalog immediate-consequence
semantics induced by \(P\).
\end{definition}

\begin{remark}
These properties are standard consequences of least-fixpoint Datalog
semantics over finite active domains
\cite{van1976semantics,ceri1989you,abiteboul1995foundations,
dantsin2001complexity}.
\end{remark}

\begin{proposition}[Basic assumptions hold for finite active-domain Datalog]
\label{prop:datalog-basic-assumptions}
For every finite active-domain Datalog instance,
Assumptions~\textup{\ref{assump:proof-system}}--\textup{\ref{assump:core-extractable}}
and Assumption~\textup{\ref{assump:finite-step-closure}} hold.
\end{proposition}

\begin{proof}
Because the program is function-free and the active domain is finite,
the set \(\mathbb{S}_O\) of ground atoms is finite.
Hence every knowledge base \(S_O\subseteq\mathbb{S}_O\) is finite and
effectively listable.

Derivability in Datalog is decidable by iterating the immediate
consequence operator \(T_P\) until the least fixpoint is reached
\cite{van1976semantics,ceri1989you,abiteboul1995foundations,
dantsin2001complexity}. Since the active domain is finite, the chain
of iterates stabilizes in finitely many steps.
Therefore membership \(s\in\Cn(\Gamma)\) is decidable for every finite
\(\Gamma\subseteq\mathbb{S}_O\) and every \(s\in\mathbb{S}_O\),
proving Assumptions~\ref{assump:proof-system} and
\ref{assump:core-extractable}.

Finally, for Assumption~\ref{assump:finite-step-closure}, one may take
the inflationary immediate-consequence operator
\[
  T_{\mathsf{PS}}(B):=B\cup T_P(B),
\]
where \(T_P\) denotes the standard one-step Datalog consequence
operator induced by \(P\). This operator is monotone, inflationary,
computable, and over a finite active domain its iterates stabilize in
finitely many steps at the least fixpoint, which equals the deductive
closure. Thus Assumption~\ref{assump:finite-step-closure} holds as
well.
\end{proof}

\subsection{A Concrete Class Ensuring Core-Disjointness}
\label{app:datalog-disjointness}

\begin{definition}[Tag-seeded Datalog class]
\label{def:tag-propagating-datalog}
Let \(A=\Atom(S_O)\).
A finite active-domain Datalog instance is called
\emph{tag-seeded} if there exists an injective map
\(\tau:A\to\mathcal D\) assigning a unique tag to each core fact,
together with a distinguished argument position in the relevant
predicates, such that:
\begin{enumerate}[label=\textup{(TS\arabic*)}]
  \item each \(a\in A\) appears with tag \(\tau(a)\), and no other
        stored fact in \(S_O\) carries the tag \(\tau(a)\);
  \item every instantiated rule propagates but never creates the
        distinguished tag: if all body facts used in a derivation
        carry the same tag \(t\), then the derived head fact also
        carries tag \(t\); moreover, a fact carrying tag \(t\) can be
        derived only from premises of which at least one already
        carries tag \(t\);
  \item for every \(a\in A\), there exists at least one fact
        \(w_a\in\Cn(S_O)\) carrying tag \(\tau(a)\) such that
        \(w_a\notin\Cn(S_O\setminus\{a\})\).
\end{enumerate}
\end{definition}

\begin{proposition}[Tag seeding implies core-disjointness]
\label{prop:datalog-core-disjoint}
Let a finite active-domain Datalog instance be tag-seeded in the
sense of Definition~\ref{def:tag-propagating-datalog}, and let the
reconstruction alphabet satisfy
\(\hat S_O\subseteq\Cn(S_O)\cap\mathbb{S}_O\).
Then Assumption~\textup{\ref{assump:core-disjoint}} holds.
\end{proposition}

\begin{proof}
Fix \(a\in A\), and let \(\hat s\in R_{\Cn}(a)\).
Then
\[
  \Cn\bigl((S_O\setminus\{a\})\cup\{\hat s\}\bigr)=\Cn(S_O).
\]
In particular, the witness fact \(w_a\) belongs to the left-hand
closure.
By \textup{(TS2)}, any derivation of \(w_a\) must preserve the tag
\(\tau(a)\).
By \textup{(TS1)}, the knowledge base \(S_O\setminus\{a\}\) contains
no stored fact carrying the tag \(\tau(a)\).
Therefore \(\hat s\) itself must carry the tag \(\tau(a)\).

Hence every element of \(R_{\Cn}(a)\) carries the tag \(\tau(a)\).
Because the tags are distinct across different core facts by
injectivity of \(\tau\), no reconstruction symbol can belong
simultaneously to \(R_{\Cn}(a_1)\) and \(R_{\Cn}(a_2)\) when
\(a_1\neq a_2\).
Thus
\[
  R_{\Cn}(a_1)\cap R_{\Cn}(a_2)=\varnothing,
  \qquad a_1\neq a_2,
\]
which is Assumption~\ref{assump:core-disjoint}.
\end{proof}

\begin{remark}[Finite-depth variants]
\label{rem:datalog-delta-disjoint}
The tagged-provenance construction above justifies
Assumption~\ref{assump:core-disjoint} for the unconstrained
zero-distortion problem.
A finite-depth analogue can also be imposed to justify
Assumption~\ref{assump:delta-core-disjoint}, but it requires tagging
the \(\delta\)-irredundant symbols themselves rather than only the
final core \(A\).
Since this extra bookkeeping is orthogonal to the main line of the
paper, we do not formalize the finite-depth syntactic condition here.
\end{remark}

\begin{remark}[Scope of the syntactic condition]
\label{rem:datalog-scope}
The tag-seeded class is only a sufficient condition for
Assumptions~\ref{assump:core-disjoint} and
\ref{assump:delta-core-disjoint}, not a necessary one.
Its role is to show that the disjointness assumptions used in the
main theorems hold on a concrete and checkable Datalog family, rather
than being purely abstract hypotheses.
\end{remark}

\section{Extra Examples and Counterexamples}
\label{app:extra-examples}

This appendix records supplementary examples illustrating the main
definitions and the limits of the theory.

\subsection{Order Dependence of the Irredundant Core}
\label{app:example-order-dependence}

\begin{example}[Order dependence]
\label{ex:order-dependence}
Let
\[
  S_O=\{p,q,r\},
\]
and suppose the proof system contains the two rules
\[
  p\leftarrow q,
  \qquad
  q\leftarrow p.
\]
Then \(p\) and \(q\) are deductively interchangeable, while \(r\) is
independent.

If the canonical order is \(p\prec q\prec r\), then \(p\) is deleted
when scanned because \(p\in\Cn(\{q,r\})\), and the resulting core is
\[
  \Atom(S_O)=\{q,r\}.
\]
If the canonical order is \(q\prec p\prec r\), then the resulting
core is
\[
  \Atom(S_O)=\{p,r\}.
\]
Thus \(\Atom(S_O)\) may depend on the canonical order, even though in
both cases the closure it generates is the same.
\end{example}

\subsection{A sufficient condition for order invariance}
\label{app:order-invariance-positive}

\begin{definition}[Essential set]
\label{def:essential-set}
The \emph{essential set} of \(S_O\) is
\[
  \Ess(S_O)
  :=
  \{s\in S_O:\ s\notin \Cn(S_O\setminus\{s\})\}.
\]
\end{definition}

\begin{proposition}[Order-robust core under essential generation]
\label{prop:essential-order-invariance}
For every finite knowledge base \(S_O\):
\begin{enumerate}[label=\textup{(\roman*)}]
  \item \(\Ess(S_O)\) depends only on \(S_O\) and \(\Cn\), not on the
        canonical order;
  \item \(\Ess(S_O)\subseteq \Atom(S_O)\) for every canonical order;
  \item if \(\Cn(\Ess(S_O))=\Cn(S_O)\), then
        \[
          \Atom(S_O)=\Ess(S_O)
        \]
        for every canonical order.
\end{enumerate}
\end{proposition}

\begin{proof}
Part~\textup{(i)} is immediate from the definition.

For~\textup{(ii)}, let \(s\in\Ess(S_O)\), and let \(B\) be the current
set when \(s\) is scanned in the deletion procedure of
Definition~\ref{def:atom-so}.
Since \(B\setminus\{s\}\subseteq S_O\setminus\{s\}\), monotonicity of
\(\Cn\) gives
\[
  \Cn(B\setminus\{s\})
  \subseteq
  \Cn(S_O\setminus\{s\}).
\]
Because \(s\notin\Cn(S_O\setminus\{s\})\), one also has
\(s\notin\Cn(B\setminus\{s\})\), so \(s\) cannot be deleted.
Hence \(s\in\Atom(S_O)\).

For~\textup{(iii)}, assume \(\Cn(\Ess(S_O))=\Cn(S_O)\) and let
\(s\in\Atom(S_O)\).
If \(s\notin\Ess(S_O)\), then
\(s\in\Cn(S_O\setminus\{s\})\), while by part~\textup{(ii)} all
essential elements survive every deletion order.
Hence, at the moment \(s\) is scanned, the current set still contains
\(\Ess(S_O)\setminus\{s\}=\Ess(S_O)\). By monotonicity,
\[
  \Cn(\Ess(S_O))
  \subseteq
  \Cn(B\setminus\{s\}).
\]
Since \(\Cn(\Ess(S_O))=\Cn(S_O)\) by assumption and \(s\in\Cn(S_O)\)
by reflexivity, it follows that
\[
  s\in\Cn(B\setminus\{s\}),
\]
which would force \(s\) to be deleted---a contradiction.
Therefore every element of \(\Atom(S_O)\) is essential, so together
with part~\textup{(ii)} one gets
\(\Atom(S_O)=\Ess(S_O)\).
\end{proof}

\begin{corollary}[Order-invariant materialized EDB/IDB stores]
\label{cor:edb-idb-order-invariant}
Suppose \(S_O\) consists of EDB facts together with materialized IDB
facts of a function-free Datalog program whose rule heads are all IDB
predicates.
Then no EDB fact is derivable from the remaining stored facts, the EDB
part generates \(\Cn(S_O)\), and therefore \(\Atom(S_O)\) equals the
stored EDB part for every canonical order.
\end{corollary}

\begin{proof}
Since all rule heads are IDB predicates, no rule derives an EDB fact.
Hence every stored EDB fact is essential.
Moreover, the stored EDB facts generate all materialized IDB facts by
construction, so the EDB subset generates \(\Cn(S_O)\).
Apply Proposition~\ref{prop:essential-order-invariance}(iii).
\end{proof}

\subsection{A Minimal Exact Zero-Distortion Example}
\label{app:example-zero-rate}

\begin{example}[One redundant consequence]
\label{ex:minimal-zero-rate}
Let
\[
  S_O=\{a,b,c\},
\]
and suppose \(c\in\Cn(\{a,b\})\), while neither \(a\) nor \(b\) is
derivable from the other two statements.
Then
\[
  A=\Atom(S_O)=\{a,b\},
  \qquad
  J=\{c\}.
\]
Assume further that \(S_O\subseteq\hat S_O\) and that
Assumption~\ref{assump:core-disjoint} holds for this instance.

If \(P_O\) is uniform on \(S_O\), then
\[
  P_A=\frac{2}{3},
  \qquad
  \pi_A(a)=\pi_A(b)=\frac{1}{2}.
\]
Hence Theorem~\ref{thm:tight-zero-rate} gives
\[
  R_{\mathrm{sem}}(0)
  =
  P_AH(\pi_A)
  =
  \frac{2}{3}\cdot 1
  =
  \frac{2}{3}\ \text{bits}.
\]
By contrast, classical zero-distortion coding under Hamming fidelity
requires
\[
  H(P_O)=\log 3
\]
bits. Thus the redundant consequence \(c\) creates a strict
deductive compression gain.
\end{example}

\subsection{A Simple Rate--Depth Example}
\label{app:example-rate-depth}

\begin{example}[A two-step shortcut]
\label{ex:rate-depth}
Let
\[
  S_O=\{a,b,d,e\},
\]
and suppose the proof system contains auxiliary statements
\(c,f\notin S_O\) and rules
\[
  c\leftarrow a,
  \qquad
  d\leftarrow c,
  \qquad
  f\leftarrow b,
  \qquad
  e\leftarrow f.
\]
Then \(d\) is derivable from \(a\) with derivation depth \(2\), and \(e\)
is derivable from \(b\) with derivation depth \(2\).
The irredundant core is
\[
  A=\{a,b\}.
\]
Assume throughout this example that \(S_O\subseteq\hat S_O\), that
Assumption~\ref{assump:atom-equals-essential} holds, and that
Assumption~\ref{assump:delta-core-disjoint} holds for
\(\delta=0,1,2\).

However, with depth budget \(\delta=1\), neither \(d\) nor \(e\)
is redundant relative to \(S_O\setminus\{d\}\) or \(S_O\setminus\{e\}\),
because the intermediate statements \(c\) and \(f\) are not stored.
Hence
\[
  A_0=A_1=S_O.
\]
At depth budget \(\delta=2\), both \(d\) and \(e\) become redundant,
so
\[
  A_2=\{a,b\}=A.
\]

If \(P_O\) is uniform on \(S_O\), then
\[
  R_{\mathrm{sem}}(0,0)=R_{\mathrm{sem}}(0,1)=H(P_O)=2
\]
bits, whereas
\[
  R_{\mathrm{sem}}(0,2)
  =
  P_AH(\pi_A)
  =
  \frac{1}{2}\cdot 1
  =
  \frac{1}{2}
\]
bit.
This illustrates how additional inference depth can sharply reduce the
effective source rate.
\end{example}

\subsection{A Confusable-Core Counterexample}
\label{app:counterexample-confusable}

\begin{example}[Failure of Assumption~\ref{assump:core-disjoint}]
\label{ex:confusable-core}
Let
\[
  S_O=\{a_1,a_2,b\},
\]
and let the proof system contain an auxiliary statement \(r\notin S_O\)
together with rules
\[
  r\leftarrow a_1,a_2,
  \qquad
  a_1\leftarrow b,r,
  \qquad
  a_2\leftarrow b,r.
\]
Neither \(a_1\) nor \(a_2\) is derivable from the remaining stored
statements alone, so both are core elements.

Now replace \(a_1\) by \(r\).
From \(\{a_2,b,r\}\), one can derive \(a_1\), and the resulting
closure equals \(\Cn(S_O)\).
Similarly, replacing \(a_2\) by \(r\) also preserves the full closure.
Therefore
\[
  r\in R_{\Cn}(a_1)\cap R_{\Cn}(a_2),
\]
so Assumption~\ref{assump:core-disjoint} fails.

This example shows that the general zero-distortion problem has a real
core-confusability phenomenon and is not captured by the disjoint case
alone.
\end{example}

\subsection{Restricted Alphabet Impossibility}
\label{app:counterexample-restricted-alphabet}

\begin{example}[No representative for a core symbol]
\label{ex:restricted-alphabet-impossible}
Let
\[
  S_O=\{a,b\},
\]
and suppose the proof system contains no nontrivial derivation rules.
Then
\[
  A=\Atom(S_O)=\{a,b\}.
\]
If the reconstruction alphabet is restricted to
\[
  V=\{a\},
\]
then \(a\) has a zero-distortion representative in \(V\), but \(b\)
does not:
\[
  R_{\Cn}^{(V)}(b)=\varnothing.
\]
Therefore zero closure distortion is impossible, and
Theorem~\ref{thm:heterogeneous-extension} yields
\[
  R_{\mathrm{sem}}^{(V)}(0)=\infty.
\]
This is the simplest instance showing that heterogeneous receivers
matter only through the availability of zero-distortion
representatives for core symbols.
\end{example}

\subsection{Numerical Illustration on Materialized Datalog Stores}
\label{app:numerical-materialized}

We briefly illustrate the zero-distortion law on a family of
function-free Datalog stores consisting of EDB facts together with a
materialized subset of derived IDB facts, a viewpoint closely related
to the database literature on materialization and materialized views
\cite{gupta1995maintenance,gupta1998materialized,chaudhuri1995optimizing}.
By Corollary~\ref{cor:edb-idb-order-invariant}, the irredundant core
is exactly the stored EDB part for every canonical order, so the
example also illustrates an order-invariant regime of the theory.

For concreteness, the synthetic instance uses three EDB predicates
\(\mathsf{connected}\), \(\mathsf{supplies}\), and \(\mathsf{produces}\),
and two IDB predicates \(\mathsf{reachable}\) and \(\mathsf{available}\),
with rules
\[
  \mathsf{reachable}(X,Y)\leftarrow \mathsf{connected}(X,Y),
  \qquad
  \mathsf{reachable}(X,Z)\leftarrow
  \mathsf{reachable}(X,Y),\,\mathsf{connected}(Y,Z),
\]
\[
  \mathsf{available}(I,L)\leftarrow
  \mathsf{produces}(S,I),\,\mathsf{supplies}(S,L),
\]
\[
  \mathsf{available}(I,L)\leftarrow
  \mathsf{produces}(S,I),\,\mathsf{supplies}(S,L_0),\,
  \mathsf{reachable}(L_0,L).
\]
The numerical values below are illustrative only and play no role in
any theorem.
For the zero-distortion formula quoted below, assume also that
\(S_O\subseteq\hat S_O\) and that
Assumption~\ref{assump:core-disjoint} holds on this synthetic family.

Consider a fixed 200-location supply-chain Datalog instance with
\(|A|=1705\) stored EDB facts.
Let \(S_O:=A\cup J_\mu\), where \(J_\mu\) is a materialized subset of
derived IDB facts of size controlled by a materialization level
\(\mu\).
Under the uniform source on \(S_O\),
Theorem~\ref{thm:tight-zero-rate} gives
\[
  R_{\mathrm{sem}}(0)
  =
  \frac{|A|}{|S_O|}\log |A|,
  \qquad
  R(0;d_H)=\log |S_O|.
\]
Table~\ref{tab:compression-gain-appendix} reports the resulting
compression ratio
\(R_{\mathrm{sem}}(0)/R(0;d_H)\).

\begin{table}[ht]
\centering
\caption{Synthetic illustration of deductive compression gain under
materialized redundant consequences (uniform source).}
\label{tab:compression-gain-appendix}
\renewcommand{\arraystretch}{1.15}
\setlength{\tabcolsep}{5pt}
\footnotesize
\begin{tabular}{rrrrr}
\toprule
\(\mu\) (\%) & \(|J|\) & \(|S_O|\) &
\(R_{\mathrm{sem}}(0)/R(0;d_H)\) &
\(\log|A|/\log|S_O|\) \\
\midrule
  0   &      0 &  1\,705 & 1.000 & 1.000 \\
 10   &  4\,340 &  6\,045 & 0.241 & 0.855 \\
 20   &  8\,680 & 10\,385 & 0.132 & 0.805 \\
 30   & 13\,020 & 14\,725 & 0.090 & 0.775 \\
 50   & 21\,700 & 23\,405 & 0.054 & 0.740 \\
 80   & 34\,720 & 36\,425 & 0.033 & 0.709 \\
100   & 43\,400 & 45\,105 & 0.026 & 0.694 \\
\bottomrule
\end{tabular}
\end{table}

The table illustrates the basic mechanism of the theory:
as more deductively redundant consequences are materialized and stored,
the irredundant core remains unchanged while the classical
zero-distortion benchmark grows with \(|S_O|\).
Accordingly, the closure-based zero-distortion rate becomes a rapidly
shrinking fraction of the classical Hamming zero-distortion rate.

\bibliographystyle{IEEEtran}
\bibliography{ref}

@book{hodges1993model,
  title={Model theory},
  author={Hodges, Wilfrid},
  year={1993},
  publisher={Cambridge university press}
}

@article{dantsin2001complexity,
  title={Complexity and expressive power of logic programming},
  author={Dantsin, Evgeny and Eiter, Thomas and Gottlob, Georg and Voronkov, Andrei},
  journal={ACM Computing Surveys (CSUR)},
  volume={33},
  number={3},
  pages={374--425},
  year={2001},
  publisher={ACM New York, NY, USA}
}

@article{shannon1948mathematical,
  title={A mathematical theory of communication},
  author={Shannon, Claude E},
  journal={Bell System Technical Journal},
  volume={27},
  number={3},
  pages={379--423},
  year={1948}
}

@book{cover2006elements,
  author={Cover, Thomas M and Thomas, Joy A},
  title={Elements of Information Theory},
  edition={2nd},
  year={2006},
  publisher={John Wiley \& Sons}
}

@inproceedings{chaudhuri1995optimizing,
  title={Optimizing queries with materialized views},
  author={Chaudhuri, Surajit and Krishnamurthy, Ravi and Potamianos, Spyros and Shim, Kyuseok},
  booktitle={IEEE International Conference on Data Engineering},
  pages={190--200},
  year={1995}
}

@article{lipkus1999proof,
  title={A proof of the triangle inequality for the Tanimoto distance},
  author={Lipkus, Alan H},
  journal={Journal of Mathematical Chemistry},
  volume={26},
  number={1},
  pages={263--265},
  year={1999},
  publisher={Springer}
}

@inproceedings{broder1997resemblance,
  title={On the resemblance and containment of documents},
  author={Broder, Andrei Z},
  booktitle={Proceedings. Compression and Complexity of SEQUENCES 1997 (Cat. No. 97TB100171)},
  pages={21--29},
  year={1997},
  organization={IEEE}
}

@article{ceri1989you,
  author={Ceri, Stefano and Gottlob, Georg and Tanca, Letizia},
  title={What you always wanted to know about {D}atalog (and never dared to ask)},
  journal={IEEE Transactions on Knowledge and Data Engineering},
  volume={1},
  number={1},
  pages={146--166},
  year={1989},
  publisher={IEEE}
}

@book{abiteboul1995foundations,
  author    = {Abiteboul, Serge and Hull, Richard and Vianu, Victor},
  title     = {Foundations of Databases},
  publisher = {Addison-Wesley},
  year      = {1995}
}

@article{gupta1995maintenance,
  title={Maintenance of Materialized Views: Problems, Techniques, and Applications},
  author={Gupta, Ashish and Mumick, Inderpal Singh},
  journal={IEEE Data Engineering Bulletin},
  volume={18},
  number={2},
  pages={3--18},
  year={1995}
}

@article{shannon1959coding,
  title={Coding theorems for a discrete source with a fidelity criterion},
  author={Shannon, Claude E and others},
  journal={IRE Nat. Conv. Rec},
  volume={4},
  number={142-163},
  pages={1},
  year={1959}
}

@book{csiszar2011information,
  title={Information theory: coding theorems for discrete memoryless systems},
  author={Csisz{\'a}r, Imre and K{\"o}rner, J{\'a}nos},
  year={2011},
  publisher={Cambridge University Press}
}

@techreport{carnap1952outline,
  author      = {Carnap, Rudolf and Bar-Hillel, Yehoshua},
  title       = {An outline of a theory of semantic information},
  institution = {Research Laboratory of Electronics, MIT},
  number      = {Technical Report 247},
  year        = {1952}
}

@article{gunduz2022beyond,
  title={Beyond transmitting bits: Context, semantics, and task-oriented communications},
  author={G{\"u}nd{\"u}z, Deniz and Qin, Zhijin and Aguerri, Inaki Estella and Dhillon, Harpreet S and Yang, Zhaohui and Yener, Aylin and Wong, Kai Kit and Chae, Chan-Byoung},
  journal={IEEE Journal on Selected Areas in Communications},
  volume={41},
  number={1},
  pages={5--41},
  year={2022},
  publisher={IEEE}
}

@article{niu2024mathematical,
  title={A mathematical theory of semantic communication},
  author={Niu, Kai and Zhang, Ping},
  journal={Journal on Communications},
  volume={45},
  number={6},
  pages={7--59},
  year={2024}
}

@article{ma2025theory,
  title={A theory for semantic channel coding with many-to-one source},
  author={Ma, Shuai and Zhang, Chuanhui and Qi, Huayan and Li, Hang and Bi, Yue and Shi, Guangming and Al-Dhahir, Naofal},
  journal={IEEE Transactions on Cognitive Communications and Networking},
  year={2025},
  publisher={IEEE}
}

@article{han2025extended,
  title={Extended Blahut-Arimoto algorithm for semantic rate-distortion function},
  author={Han, Y. and Liu, Y. and Sun, Y. and Niu, K. and Ma, N. and Cui, S. and Zhang, P.},
  journal={Entropy},
  volume={27},
  number={6},
  pages={651},
  year={2025}
}

@article{wyner1976rate,
  author  = {Wyner, Aaron D. and Ziv, Jacob},
  title   = {The rate-distortion function for source coding
             with side information at the decoder},
  journal = {IEEE Trans.\ Inform.\ Theory},
  volume  = {22},
  number  = {1},
  pages   = {1--10},
  year    = {1976}
}

@article{shannon1956zero,
  title={The zero error capacity of a noisy channel},
  author={Shannon, C},
  journal={IRE Transactions on Information Theory},
  volume={2},
  number={3},
  pages={8--19},
  year={1956}
}

@article{witsenhausen1976zero,
  title={The zero-error side information problem and chromatic numbers (Corresp.)},
  author={Witsenhausen, H},
  journal={IEEE Transactions on Information Theory},
  volume={22},
  number={5},
  pages={592--593},
  year={1976}
}

@article{alon1996source,
  title={Source coding and graph entropies},
  author={Alon, N and Orlitsky, A},
  journal={IEEE Transactions on Information Theory},
  volume={42},
  number={5},
  pages={1329--1339},
  year={1996}
}

@article{orlitsky2001coding,
  title={Coding for Computing},
  author={Orlitsky, Alon and Roche, James R},
  journal={IEEE Transactions on Information Theory},
  volume={47},
  number={3},
  pages={903--917},
  year={2001}
}

@article{van1976semantics,
  title={The semantics of predicate logic as a programming language},
  author={Van Emden, Maarten H and Kowalski, Robert A},
  journal={Journal of the ACM (JACM)},
  volume={23},
  number={4},
  pages={733--742},
  year={1976},
  publisher={ACM New York, NY, USA}
}

@article{slepian1973noiseless,
  title={Noiseless Coding of Correlated Information Sources},
  author={SLEPIAN, DAVID and WOLF, JACK K},
  journal={IEEE TRANSACTIONS ON INFORMATION THEORY},
  volume={19},
  number={4},
  pages={471},
  year={1973}
}

@inproceedings{liu2022indirect,
  author    = {Liu, Jiakun and Shao, Shuo and Zhang, Wenyi and Poor, H. Vincent},
  title     = {An indirect rate-distortion characterization for semantic sources: General model and the case of Gaussian observation},
  booktitle = {IEEE Transactions on Communications},
  volume    = {70},
  number    = {9},
  pages     = {5946--5959},
  year      = {2022},
  doi       = {10.1109/TCOMM.2022.3192049}
}

@article{zhao2025semantic,
  title={Semantic Rate-Distortion Theory with Applications},
  author={Zhao, Yi-Qun and Ma, Zhi-Ming and Li, Geoffrey Ye and Yuan, Shuai and Ye, Tong and Zhou, Chuan},
  journal={arXiv preprint arXiv:2509.10061},
  year={2025}
}

@article{salehkalaiabar2024rate,
  author    = {Salehkalaiabar, Sadaf and Chen, Jun and Khisti, Ashish and Yu, Wei},
  title     = {Rate-distortion-perception tradeoff based on the conditional-distribution perception measure},
  journal   = {IEEE Transactions on Information Theory},
  volume    = {70},
  number    = {12},
  pages     = {8432--8454},
  year      = {2024},
  doi       = {10.1109/TIT.2024.3434979}
}

@inproceedings{hu2019datalog,
  author    = {Hu, Pan and Urbani, Jacopo and Motik, Boris and Horrocks, Ian},
  title     = {Datalog reasoning over compressed {RDF} knowledge bases},
  booktitle = {Proceedings of the 28th ACM International Conference on Information and Knowledge Management (CIKM)},
  pages     = {2065--2068},
  year      = {2019},
  doi       = {10.1145/3357384.3358147}
}

@book{gupta1998materialized,
  title={Materialized views: techniques, implementations, and applications},
  author={Gupta, Ashish and Mumick, Inderpal Singh and others},
  year={1998},
  publisher={MIT press Cambridge, MA}
}

@article{korner1986fredman,
  title={Fredman--Koml{\'o}s bounds and information theory},
  author={K{\"o}rner, J{\'e}nos},
  journal={SIAM Journal on Algebraic Discrete Methods},
  volume={7},
  number={4},
  pages={560--570},
  year={1986},
  publisher={SIAM}
}

@inproceedings{korner1971coding,
  title={Coding of an information source having ambiguous alphabet and the entropy of graphs.},
  author={Korner, Janos and others},
  booktitle={6th Prague conference on Information Theory, etc.},
  pages={411--425},
  year={1971},
  organization={Academia, Prague}
}

\vfill
\end{document}